\documentclass[aps, pra, 10pt, twocolumn, superscriptaddress, amsmath, amssymb, notitlepage, letterpaper]{revtex4-1}

\usepackage{pgfplots}
\pgfplotsset{compat=newest}
\usepackage{tikz}
\usepackage[strict]{revquantum}
\usepackage{pifont}

\newlength\figureheight
\newlength\figurewidth
\usepackage{soul}

\usepackage{kbordermatrix}
\renewcommand{\kbldelim}{(}
\renewcommand{\kbrdelim}{)}

\usetikzlibrary{shapes.geometric,arrows}

\tikzset{
     causal node/.style={
        draw, fill=white, circle
    },
    causal edge/.style={
        ->,
        thick,
        >=stealth
    },
    marginalized/.style={
        fill=cud-reddish-purple,
        text=white,
        draw=white
    },
    conditioned/.style={
        fill=cud-bluish-green,
        text=white,
        draw=white
    },
    parent/.style={},
    target/.style={},
    penalized/.style={
        draw=cud-vermillion
    }
}
\tikzstyle{startstop}=[rectangle,rounded corners,minimum width=3cm,minimum height=1cm,text centered, draw=black, fill=cud-vermillion!30]
\tikzstyle{io}=[trapezium, trapezium left angle=70,trapezium right angle=110,minimum width=3cm,minimum height=1cm,text centered,draw=black,fill=cud-bluish-green!80]
\tikzstyle{process}=[rectangle,minimum width=3cm,minimum height=1cm,text centered,draw=black,fill=cud-reddish-purple!30]
\tikzstyle{additionalInfo}=[rectangle,rounded corners,anchor=west,text width=3cm,minimum width=3cm,minimum height=1cm,text centered,draw=black,fill=cud-bluish-green!10]
\tikzstyle{decision}=[diamond,minimum width=3cm,minimum height=1cm,text centered,draw=black,fill=orange!70]
\tikzstyle{mdecision}=[diamond,minimum width=0.2cm,minimum height=1cm,text centered,draw=black,fill=orange!70]
\tikzstyle{arrow} = [thick,->,>=stealth]

\algnewcommand{\LineComment}[1]{\State \(\triangleright\) #1}

\algblockdefx[Switch]{Switch}{EndSwitch}[1]{\textbf{switch} #1}[0]{\textbf{end switch}}
\algloopdefx[Case]{Case}[1]{\textbf{case} #1}

\usepackage[caption=false]{subfig}

\newcommand{\figurefolder}{figures}

\newrm{pa}
\newrm{obs}
\newrm{gen}
\newrm{offspring}
\newoperator{Uni}

\newcommand{\maxz}[1]{\operatorname*{max^*\!}_{#1}}
\newcommand{\minz}[1]{\operatorname*{min^*\!}_{#1}}

\newaffil{RMIT}{Quantum Photonics Laboratory, School of Engineering, RMIT University, Melbourne, Victoria 3000, Australia}
\newaffil{ITP}{Institute for Theoretical Physics, University of Cologne, Germany}
\newaffil{EQ}{authors contributed equally}

\begin{document}


\title{Explaining quantum correlations through evolution of causal models}
\date{\today}

\author{Robin Harper}
\thanks{These authors contributed equally to this work.}
\affilEQUS
\affilUSydPhys

\author{Robert J. Chapman}
\thanks{These authors contributed equally to this work.}
\affilUSydPhys
\affilRMIT

\author{Christopher Ferrie}
\affilEQUS
\affilUSydPhys

\author{Christopher Granade}
\affilEQUS
\affilUSydPhys

\author{Richard Kueng}
\affilITP

\author{Daniel Naoumenko}
\affilUSydPhys

\author{Steven T.\ Flammia}
\affilEQUS
\affilUSydPhys

\author{Alberto Peruzzo}
\affilUSydPhys
\affilRMIT

\begin{abstract}
We propose a framework for the systematic and quantitative generalization of Bell's theorem using causal networks. We first consider the multi-objective optimization problem of matching observed data while minimizing the causal effect of nonlocal variables and prove an inequality for the optimal region that both strengthens and generalizes Bell's theorem. 
To solve the optimization problem (rather than simply bound it), we develop a novel genetic algorithm treating as individuals causal networks. By applying our algorithm to a photonic Bell experiment, we demonstrate the trade-off between the quantitative relaxation of one or more local causality assumptions and the ability of data to match quantum correlations.
\end{abstract}

\maketitle



While it seems conceptually obvious that causality lies at the heart of physics, its exact nature has been the subject of constant debate.
The fundamental implications of quantum theory shed new light on this debate.
It is thought these implications may lead to new insights into the foundations of quantum theory, and possibly even quantum theories of gravity \cite{Leifer2007Conditional,Leifer2008Quantum,Fritz2012Beyond,Fitzsimons2013Quantum,Leifer2013Formulating,Leifer2014A,Brukner2014Quantum,Cavalcanti2014On,Wood2015The,Ried2015A}.

These realizations have their roots in the Einstein-Podolski-Rosen thought experiment \cite{Einstein1935Can} and the fundamental theorems of \citet{Bell1964On} and of \citet{Kochen1967The}.  A cornerstone of modern physics, Bell's theorem, rigorously excludes classical concepts of causality. 
Roughly speaking Bell's theorem states that the following concepts are mutually inconsistent: (1) reality; (2) locality; (3) measurement independence; and (4) quantum mechanics.

In philosophical discussions, typically one rejects (1) or (2), which together are often referred to as \emph{local causality}, though the other options have been considered as well.  In studies with an operational bent, however, one often considers relaxations of (2) or (3) which is what we concern ourselves with here. These relaxations have been addressed from different perspectives, but only regarding specific causal influences in isolation \cite{Toner2003Communication,Barrett2011How, Hall2010Complementary, Hall2010Local, Hall2011Relaxed, Koh2012Effects, Banik2013Lack, Thinh2013Bell,Putz2014Arbitrarily,Maxwell2014Bell}, whereas here we wish to study all possible relaxations of the causal assumptions implied by (2) and (3) simultaneously.

The framework of \emph{causal networks} \cite{Pearl2009Causality,Sprites2001Causation} is wildly successful within the field of machine learning and has led some physicists to utilize them to elucidate the tension between causality and Bell's theorem.  Recently, Wood and Spekkens have shown that existing principles behind causal discovery algorithms (namely, the absence of fine tuning) still cannot be reconciled with entanglement induced quantum correlations \emph{even if one admits nonlocal models} \cite{Wood2015The}.  However, such results only hold for the exact distributions, and would not necessarily apply to experimental data due to measurement noise, or a relaxation of the demand of reproducing exactly the quantum correlations.  Clearly, the further away from the quantum correlations one is allowed to stray, the more likely a locally causal model can be found.


\begin{figure}[b!]
    \begin{centering}
        \includegraphics[width=1\columnwidth]{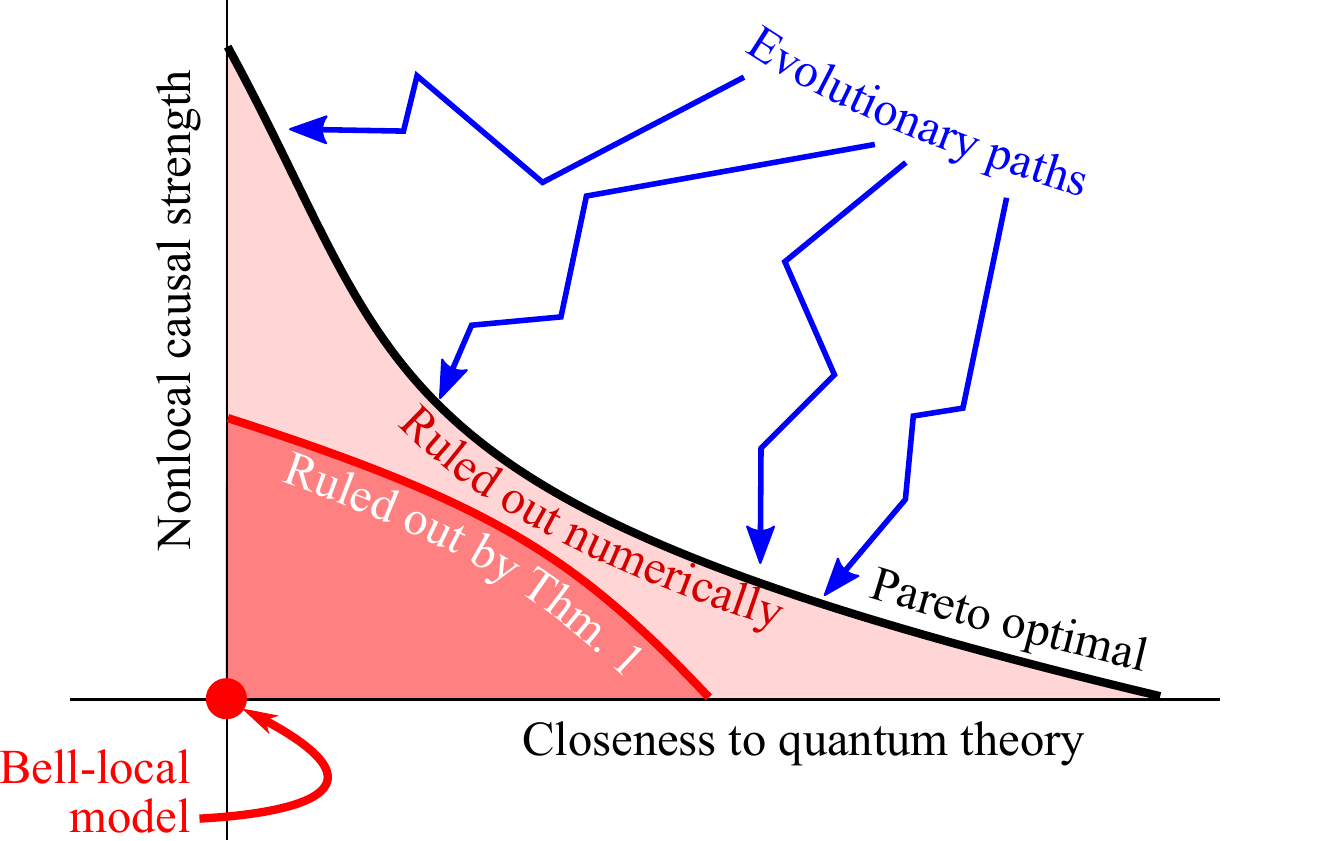}
     \caption{ \label{fig:cartoon}
        A sketch of the concept of Pareto optimality for demarcating the boundary between local causality and quantum correlations. In this picture, Bell's theorem rules out the origin only. Our results rule out an entire region of possible models in the presence of relaxations of Bell's assumptions. We rule out this region both rigorously with \autoref{thm:bounds} and numerically with the evolutionary algorithm that we developed specifically for this task.
    } \end{centering}
\end{figure}

Here we propose a framework for systematic and quantitative generalizations of Bell's theorem by using causal networks. The idea, depicted in \autoref{fig:cartoon}, is to consider the multi-objective optimization problem of matching the observed data from an experiment while minimizing the causal effect of nonlocal variables. It is in this sense of matching experimental data that we are explaining the quantum correlations. Our first contribution is a rigorous lower bound for this optimization problem, demonstrating a generalization of Bell's theorem. \autoref{thm:bounds} below establishes that there must exist a tradeoff between the goodness of fit to experimental data and the quantitative amount of causal influence for \emph{any} model.

This theorem rules out a portion of the space allowed by this new framework, but the bounds are not tight. To solve the optimization problem, and hence numerically find the \emph{optimal} bounds, we develop a type of genetic algorithm called a multi-objective evolutionary algorithm (MOEA) to quantify the relaxations necessary to reproduce the data generated by experiments on entangled quantum systems \cite{Hensen2015,PhysRevLett.115.250401,PhysRevLett.115.250402}. Our genetic algorithm treats as individuals causal networks and we develop genetic operators which represent the evolution of these networks. 
By applying our algorithm to a photonic Bell experiment, we show that the tradeoff between the quantitative relaxation of one or more local causality assumptions and the ability to match quantum correlations appears linear.

The outline of the paper is as follows. In \autoref{sec:causal} and \autoref{sec:bell} we set out the background of the causal models we use and the mathematics required to convert a probability distribution into a fitness function. In \autoref{sec:bounds} we provide analytic bounds on causal influence. In \autoref{sec:experiment} we describe the experiment that provided the input to the algorithm. \autoref{sec:num_res} briefly describes the result of applying the genetic algorithm to the experimental data. \autoref{sec:ea} describes the process by which we convert the problem into one that can be explored using evolutionary operators and details the construction of the algorithm. We conclude in \autoref{sec:discussion} with a discussion.

\section{Causal models for Bell Experiments}\label{sec:causal}

The formalism of causal models allows us to quantify the relaxations necessary to avoid the contradiction in Bell's theorem and, more importantly, explore the trade-offs necessary in minimizing the amount by which the assumptions are violated.  Building off the work of Chaves \emph{et al} \cite{Chaves2015A}, we make all this concrete through a quantification of the relaxation of each assumption in the context of causal models.  The task of minimizing the amount of the relaxation is a multi-objective optimization problem.  Bell's theorem is recast as the statement that all objectives cannot be simultaneously minimized.  We explore the trade-offs through the concept of Pareto optimality.

\tikzset{      
       visible node/.style={
               draw=white, 
               fill=cud-reddish-purple, 
               circle,
               minimum size = 0.6cm,
               text=white, scale=1.5
           },
       hidden node/.style={
            draw=white, fill=cud-bluish-green, circle,minimum size = 0.6cm,scale=1.5
       },
        violation edge/.style={
            ->,dotted,
            thick,
            >=stealth
        }
    }
\begin{figure}\scalebox{.55}{
    \begin{centering}
        \subfloat[]{
        \label{fig:bellgraphA}
        \includegraphics{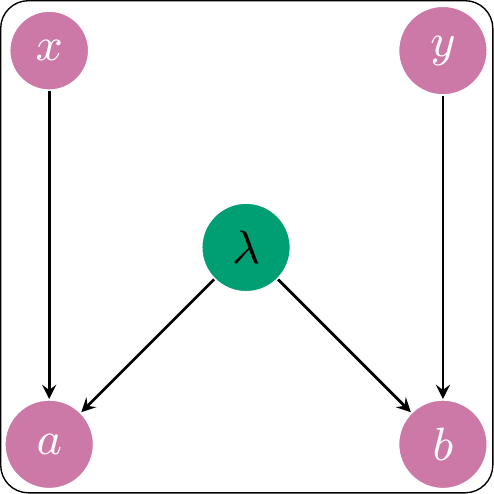}
%

    }   
        \subfloat[
    ]{
    \label{fig:bellgraphB}
       \includegraphics{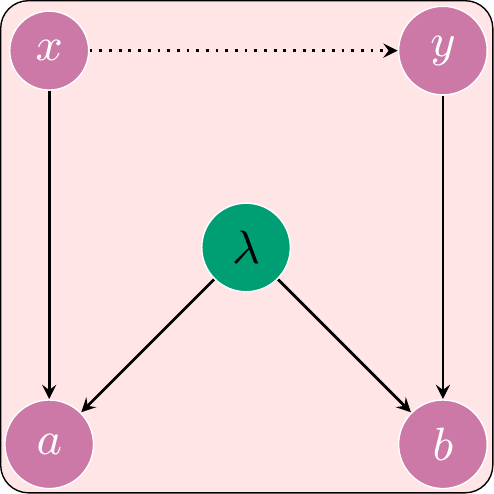}
       
       

        }
     \subfloat[
    ]{
     \label{fig:bellgraphC}
     \includegraphics{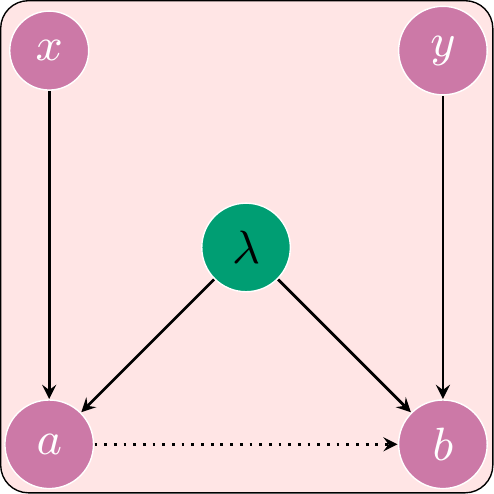}
       

        }

    \end{centering}
    }
    \caption{
        \label{fig:all_bell_dags}
        Causal networks for Bell-type experiments.  On the left is the local hidden variable model, which respects the assumptions going into Bell's famous no-go theorem.  Such a model cannot account for certain correlations obtained from measuring entangled particles.  The graph in the middle contains a causal link between the measurement settings.  Such a model exploits the detection loophole and violates measurement independence. Finally, on the right is a superluminal model which contains a causal link between the outcomes of the experiments.
    }
\end{figure}

The prototypical ``Bell experiment'' has two distant parties, often named Alice and Bob.
We suppose that Alice and Bob each have devices with binary measurement
settings, respectively labeled $x$ and $y$.
Conditioned on these measurement settings, their devices also record binary events, labeled $a$ (Alice) and $b$ (Bob).  Suppose it is empirically observed that $a$ and $b$ are correlated.  Bell defined a locally causal model of such correlations as follows: there exists a ``hidden variable'' $\lambda$ which is the common cause both of $a$ conditioned on $x$, and of $b$ conditioned on $y$. We write these random variables as $a\mathrel{|}x$ and $b\mathrel{|}y$, respectively.
Formally, the general conditional distribution is assumed to satisfy
\begin{equation}\label{eq:local}
\Pr(a,b|x,y,\lambda) = \Pr(a|x,\lambda)\Pr(b|y,\lambda).
\end{equation} 
Moreover, it is assumed that the choices of settings can be made such that each of $x$ and $y$ can be set independently of the hidden variable $\lambda$,
\begin{equation}\label{eq:meas}
\Pr(x,y|\lambda) = \Pr(x|\lambda)\Pr(y|\lambda) = \Pr(x) \Pr(y).
\end{equation} 
Such an assumption is often motivated by the injection of randomness into the measurement settings or the free-will of Alice and Bob.  Bell's theorem can be stated succinctly as follows: the conditional distributions describing the outcomes of some experiments on quantum systems cannot be factorized as in Eqs. \eqref{eq:local} and \eqref{eq:meas}.   

 A causal network is a directed acyclic graph with nodes representing random variables and edges denoting causal relationships between variables.  The defining feature of such networks is the factorization of joint probabilities.  Generally, suppose we have nodes $\{x_0,x_1,\ldots,x_\text{K}\}$, each of which
represents a random variable in our model. We will assume that each such random variable is discrete,
and without loss of generality, will assume integer labels $x_i \in \{0, \dots, \dim x_i - 1\}$ for its possible values.
The edges in the causal network of these variables are defined such that
\begin{equation}\label{eq:def_dag}
    \Pr(x_0,x_1,\ldots,x_\text{K}) = \prod_{i=0}^\text{K} \Pr(x_i|{\rm pa}_i),
\end{equation}
where $\pa_i$ denotes the parents of node $i$.

Take, for example, the causal network in \autoref{fig:bellgraphA}. In general, we can decompose
the joint distribution $\Pr(a, b, x, y, \lambda)$ in terms of conditional distributions as
\begin{equation}
\begin{split}
    &\Pr(a,b,x,y,\lambda) = \\
    &\Pr(a|b,x,y,\lambda)\Pr(b|x,y,\lambda)\Pr(x|y,\lambda)\Pr(y|\lambda)\Pr(\lambda).
 \end{split}
\end{equation}
Using the causal network to eliminate conditionals, \autoref{eq:def_dag} implies
\begin{equation}
\begin{split}
&\Pr(a,b,x,y,\lambda) =  \\
&\qquad\Pr(a|x,\lambda)\Pr(b|y,\lambda)\Pr(x)\Pr(y)\Pr(\lambda),
\end{split}
\end{equation}
which are identical to Bell's assumptions on local hidden variable models.  Thus, Bell's theorem is equivalent to the statement that certain quantum correlations cannot be realized by the causal network in \autoref{fig:bellgraphA}.

\section{Relaxing Bell's assumptions}\label{sec:bell}

It is known that quantum mechanical correlation arising in a Bell-type experiment \emph{can} however be explained by adding a new causal link to the local hidden variable network \cite{Barrett2011How, Hall2011Relaxed}.  Two examples are shown in \autoref{fig:all_bell_dags}.  In many practical cases, these causal links are not entirely unphysical from the standpoint of respecting relativity and free-will, for example.  The reason being that experiments do not actually conform to the exact assumptions Bell made---there are noisy detectors, non-random number generation, losses, inability to space-like separate ``Alice'' and ``Bob,'' and so on.  When this is the case, such causal models are said to be exploiting \emph{loopholes}.  

In \autoref{fig:bellgraphB}, a causal model that allows correlations between the measurement settings is shown.  In the same spirit, we could have had either $x$ or $y$ be causally dependent on $\lambda$ or another hidden variable.  Such models are often called \emph{superdeterministic} and are ruled out by the assumption that Alice and Bob are not colluding and have free-will or access to   independent randomness.  If the experiment only approximately satisfies these assumptions---perhaps due to low detection efficiency---one can still model the data with a local hidden variable said to be exploiting the \emph{detection loophole} \cite{Pearle1970Hidden}.  The question of quantifying the amount of independence of the measurement settings necessary has been addressed from multiple perspectives and has practical quantum cryptographic consequences \cite{Barrett2011How, Hall2010Complementary, Hall2010Local, Hall2011Relaxed, Koh2012Effects, Banik2013Lack, Thinh2013Bell,Putz2014Arbitrarily}.

In \autoref{fig:bellgraphC}, a causal model which allows correlations between the measurement outcomes is shown.  
This is, and similar models are, called \emph{nonlocal} and could potentially even allow for superluminal signaling. 
A quintessential example of a nonlocal model which reproduces the predictions of quantum theory is Bohmian mechanics.  Toner and Bacon studied the amount of nonlocality necessary to simulate quantum correlation in the context of classical communication costs \cite{Toner2003Communication, Maxwell2014Bell}, while Wolf has expressed nonlocality in terms of the compressibility of experimental observations \cite{wolf_nonlocality_2015}.

The current studies, mentioned above, quantifying the relaxations of the causal assumption necessary to replicate quantum correlations are rather disjoint.  Recently, Chaves \emph{et al.} placed the question in context of causal networks and found that some measures of these relaxations can be cast as efficiently solvable linear programs \cite{Chaves2015A}.  
We build on this idea and consider a completely abstract framework amenable to any set of random variables using a single measure of the causal influence of one variable on another. 
This allows us to consider \textit{all possible relaxations simultaneously} and thus explore the trade-offs necessary to simulate quantum correlations with hidden variable models.

We will now state our model more technically.  For consistency we formulate the problem in the context of the two-party Bell experiment, but we emphasize that this approach generalizes in an obvious way to any set of random variables.  A model, $M$, is specified by a joint distribution
\begin{equation}
\Pr(a,b,x,y,\lambda|M).
\end{equation}
We label the empirical frequencies $F(a,b,x,y)$ and denote the total variational distance (TVD) of a model to these frequencies by
\begin{equation}
\operatorname{TVD}(M) = \|\Pr(a,b,x,y|M) - F(a,b,x,y) \|_1 \label{eq:TVD},
\end{equation}
where the vector being normed is labeled by $(a,b,x,y)$.
Here the 1-norm of a vector $x$ is simply $\|x\|_1 = \sum_i |x_i|$.

\begin{widetext}
The causal influence is defined for a \emph{general} graph as follows:
\begin{equation}
  \label{eq:causal-influence-general}
  C_{x_i\to x_j}(M) \defeq \maxz{x_i, x'_i, \pa_j} \|\Pr(x_j|x'_i,\pa_j\setminus\pa_j^2,M)-\Pr(x_j|x_i,\pa_j\setminus\pa_j^2,M)\|_1,
\end{equation}
where $\pa_j\setminus\pa_j^2$ is the set of parents of $x_j$ that are not also grandparents of $x_j$,
and where $\maxz{}$ indicates that the maximization over $x_i$, $x'_i$ and $\pa_j$ is restricted
to feasible assignments. That is, the maximization does not consider assignments outside the support
of $M$.
In words, the causal influence is non-zero when changing $x_i$ leads to a change in $x_j$. It is quantified by maximizing over latent variables of the target that are not also latent variables of the control.
\end{widetext}

For example, if we want to minimize the causal influence between two variables $a$ to $b$ in \autoref{fig:bellgraphC} we consider
\begin{equation}
  \label{eq:causal-influence-defn}
  C_{a\to b}(M) \defeq \maxz{a,a', y} \|\Pr(b|a,y,M)-\Pr(b|a',y,M)\|_1.
\end{equation}
We include the conditions $\Pr(a), \Pr(a') \ne 0$ to prevent the causal influence
being maximized by an assignment outside the support of the random variable
$A$; the maximization should be taken over all feasible assignments.

Intuitively, this definition represents how distinguishable the different settings
of $a$ are when viewed through measurements of $b$.
That is, if $a$ does not causally affect $b$, then
it is not possible for a change in $a$ to be detectable through $b$ alone.
We adopt this definition in lieu of the traditional approach of using \emph{interventions},
wherein an external agent imposes a particular value of $a$ while holding all else fixed,
effectively cutting out any causal links incident on $a$ other than one originating from
the experimentalist themselves. 
Though some novel experiments have been performed using intervention to reason about quantum mechanics \cite{White2016}, we cannot intervene on quantum mechanical models in general, such that we must instead maximize over conditions for the experiment, here represented by the maximization over $a$ and $y$.

The task then is to find a model $M$ which minimizes TVD and $C_{\alpha \to \beta}$ for each $\alpha \to \beta$ ruled out by local causality and measurement independence.  If the empirical frequencies contain some causal dependence between two variables, then either the model must also contain such causal dependence or the observed frequencies from the model must be different from the empirical frequencies. Perhaps interestingly, one might be able to ``trade'' unwanted causal influence between one pair of variables for another, while maintaining the same TVD.  Thus, the problem of determining ``how much'' relaxation of Bell's causal assumptions is necessary to match an empirical observed frequency becomes much more interesting and nuanced.  

Suppose two models $M_1$ and $M_2$ both match the data equally well---i.e. $\operatorname{TVD}(M_1)=\operatorname{TVD}(M_2)$--- but $M_2$ has some unwanted casual influence $a\to b$, say, and $M_1$ does not---that is, $0=C_{a\to b}(M_1) < C_{a\to b}(M_2)$.  Clearly, $M_1$ is preferred and we say $M_2$ is \emph{dominated} by $M_1$.  For many objectives, the situation is more complex but can be handled by the concept of Pareto optimality.

Let $\mathcal M$ be the set of all models.  Let each model's fitness be represented by the function $f:\mathcal M \to \mathbb R^n$, where $n$ is number of objectives.  Define the partial order $\prec$ as follows:
\begin{equation}
M \succcurlyeq M' \Leftrightarrow f(M)_k \geq f(M')_k, 
\end{equation}
for all $k\in \{0,n-1\}$.  If $M \succcurlyeq M'$, we say $M'$ dominates $M$ (or is equivalent to $M$, if $M' \succcurlyeq M$ holds as well). The set $\mathcal P\subset \mathcal M$ of \emph{Pareto optimal} models is now defined as follows:
\begin{equation}
\mathcal P = \{M \in \mathcal M : \{M' \in \mathcal M : M \succcurlyeq M', M' \neq M\} = \emptyset \}. 
\end{equation} 
This says that a model is Pareto optimal if the set of other models which dominate it is empty.  In other words, the Pareto optimal is the set of non-dominated models.

\section{Analytical bounds} \label{sec:bounds}

In this section, we provide analytical bounds which relate the amount of causal influences exhibited by any model $M$ to its agreement with the empirical frequencies $F(a,b,x,y)$.
For the sake of simplicity, we restrict ourselves to analyzing the causal influence between the variables $a$ and $b$---see \autoref{fig:bellgraphC}. 
However, we emphasize that analogous statements are valid for causal influences between any two variables.
For the variables $a$ and $b$, the empirical frequencies themselves admit a causal influence
\begin{equation}
    C_{a \to b}(F) = \maxz{a,a',y} \| F(b|a, y) - F(b|a', y) \|_1
    \label{eq:data_causal_effect}
\end{equation}
which is defined in complete analogy to \autoref{eq:causal-influence-defn}. 
To state our theorem, we must define two more quantities. 
Let $\mathcal{M}_\tau = \mathcal{M}_\tau(F)$ be the set of models having $\TVD(M) \le \tau$ with respect to the empirical frequencies $F$, and denote by $f^\star = \min_a F(a)$ the minimum empirical marginal frequency. 

\begin{theorem}\label{thm:bounds}
For all models $M\in \mathcal{M}_\tau$ and $\tau < 2 f^\star$, 
\begin{equation}
|C_{a \to b} (F) - C_{a \to b} (M)| \le \frac{2 \tau (4 f^\star-\tau )}{f^\star (2 f^\star-\tau )}\,.
\label{eq:bound}
\end{equation}
\end{theorem}

We point out that the bound \eqref{eq:bound} becomes loose and eventually diverges if the minimum empirical marginal frequency $f^\star$ approaches zero or if the $\TVD$ of the class of models becomes too large relative to $f^\star$.

The proof of \autoref{thm:bounds} can be found in \hyperref[app:bounds_derivation]{Appendix A}.

\section{Bell experiment and data}\label{sec:experiment}

\begin{figure}[t!]
\centering
\includegraphics[width=1\columnwidth]{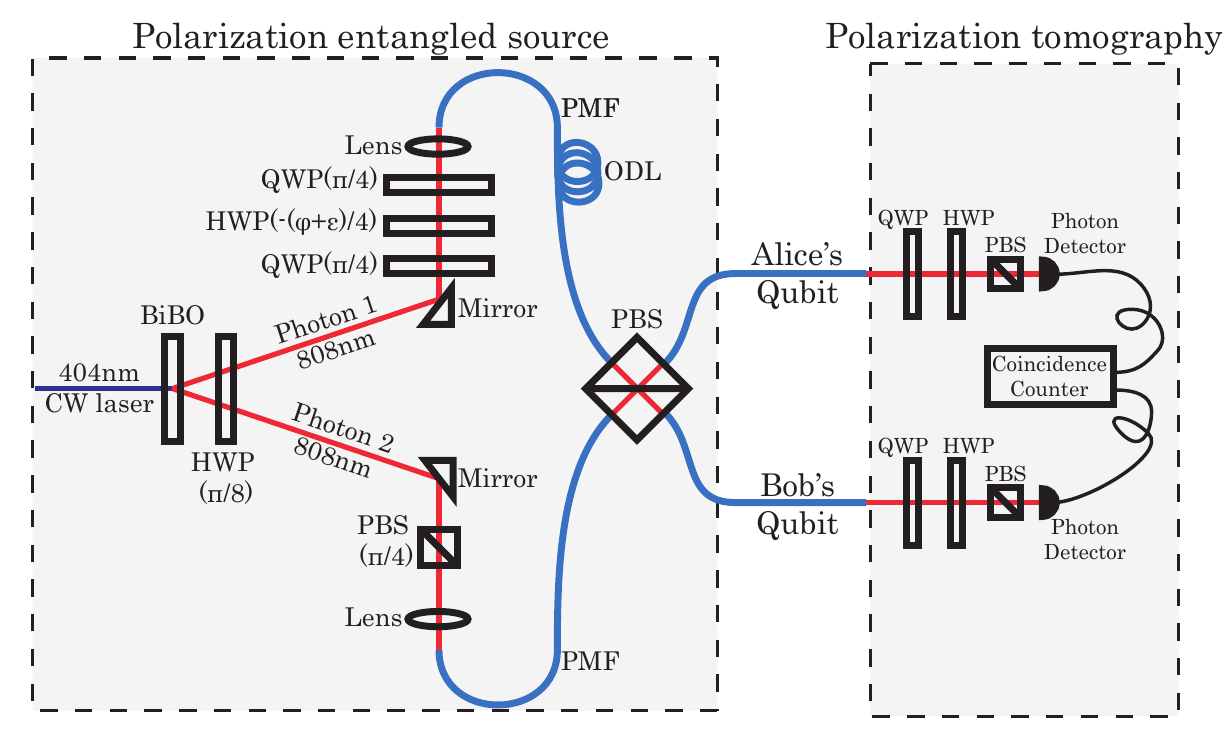}
\caption{Photon paris at 808 nm are emitted via spontaneous parametric down-conversion from a 404 nm pumped bismuth triborate (BiB0) crystal and we prepare polarization enabled photons, labelled ``Alice's Qubit'' and ``Bob's Qubit''. The components used are: half wave plates (HWP), quarter wave plates (QWP), polarizing beam splitters (PBS), polarization maintaining fibers (PMF) and an optical delay line (ODL). Once the state is prepared a polarization tomography setup enables projection of each qubit onto any pure polarization state, which is sufficient to perform two-qubit tomography.}
\label{fig:schematic}
\end{figure}

As input for the MOEA we use data from a polarization photonic Bell experiment, shown in \autoref{fig:schematic} \cite{chapman_experimental_2016}.
Indistinguishable horizontally polarized ($\ket{H}$) photon pairs are generated via type-1 spontaneous parametric down-conversion. 
Both polarization qubits are rotated into a diagonal state  $\tfrac{1}{\sqrt{2}}(\ket{H}+\ket{V})$ by a half wave plate (HWP) with fast axis at $\tfrac{\pi}{8}$ from vertical, where $\ket{V}$ denotes vertical polarization.
A polarization phase rotation is applied to photon 1 by two quarter wave plates (QWPs) and a HWP, while photon 2 has its state optimized by a polarizing beamsplitter (PBS). 
Both photons are collected in polarization maintaining optical fiber (PMF) and are incident on the two input faces of a fibre-coupled PBS, which transmits $\ket{H}$ and reflects $\ket{V}$, preparing Alice's and Bob's qubits.
The configuration of the optical fibres results in a $\sigma_x$ operation applied to Alice's qubit.
By measuring in the coincidence basis, we post-select the state
\begin{equation}
	\rho = \frac{1+\gamma}{2}\ket{\Phi^+}\bra{\Phi^+} + \frac{1-\gamma}{2}\ket{\Phi^-}\bra{\Phi^-},
\label{eq:gamma}
\end{equation}
where $\ket{\Phi^{\pm}}$ are the Bell states $\tfrac{1}{\sqrt{2}}(\ket{H_AV_B}\pm\ket{V_AH_B})$ with subscript $A$ ($B$) corresponding to Alice's (Bob's) qubit. 
The parameter $\gamma$ defines the coherence of the state which depends on the overlap of the two photons after the fibre-coupled PBS and is controlled by the optical delay line (ODL). 
The state prepared when $\gamma=1$ is a maximally entangled Bell state $\ket{\Phi^+}$ and when $\gamma=0$ is an incoherent mixture $\tfrac{1}{2}(\ket{H_AV_B}\bra{H_AV_B}+\ket{V_AH_B}\bra{V_AH_B})$.
The polarization tomography setup in \autoref{fig:schematic} enables projection onto any pure state and can be used for two-qubit state tomography \cite{james_measurement_2001}.
The photons are detected with silicon avalanche photo-diodes and coincidence counts recorded by a timing card.

\begin{figure}[t!]
\centering
\includegraphics[width=1\columnwidth]{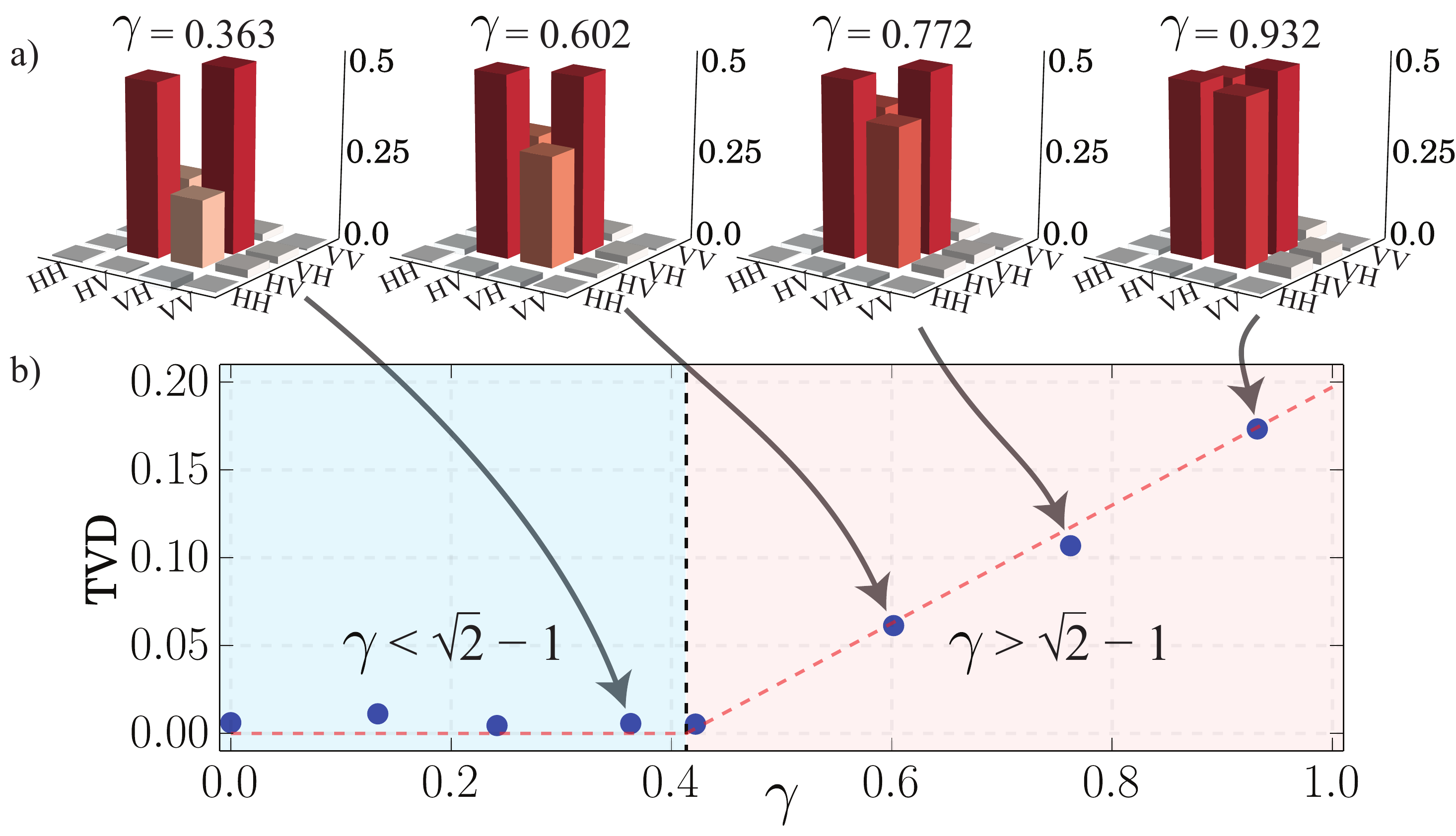}
\caption{(a) The real component of experimentally measured density matrices for a range $\gamma$ values measured via two-qubit state tomography. (b) the results of the MOEA (run as an EA) on a local model (see \autoref{fig:bellgraphA}) to generate the best achievable TVD for various values of $\gamma$ from a Bell-like experiment described in \autoref{sec:experiment}. As theoretically predicted above a certain threshold ($\gamma = \sqrt{2}-1$) the local model can no longer explain the measurement results with zero TVD. This threshold corresponds to violating the CHSH inequality in \autoref{eq:S}. The linear increase in TVD corresponds to the linear increase in $S$ as discussed in \autoref{sec:experiment}.}
\label{fig:gammagraph}
\end{figure}

Our input for the MOEA is a normalized frequency distribution $F(a,b,x,y)$ across binary measurement settings for Alice ($x=\{x_1,x_2\}$) and Bob ($y=\{y_1,y_2\}$), and binary measurement outcomes $a = \{\ket{H_A},\ket{V_A}\}$ and $b = \{\ket{H_B},\ket{V_B}\}$ respectively. 
The measurement settings are controlled by wave plate angles in the tomography and the measurement outcome is the collapse of the state onto one of the four basis state $\ket{H_AH_B}$, $\ket{H_AV_B}$, $\ket{V_AH_B}$ or $\ket{V_AV_B}$.
A single measurement is the number of photon pairs recorded for a fixed integration time and can be written as $N_{ab}^{xy} = \mathcal N \tau \bra{ab}U^{xy}\rho U^{xy\dagger}\ket{ab}$ for measurement settings $x,y$ and measurement outcomes $a,b$. $\mathcal N$ is the total photon flux, $\tau$ is the integration time and $U^{xy}$ is the operation of the wave plates.
We calculate $F(a,b,x,y)$ by measuring all combinations of $x$, $y$, $a$ and $b$, and normalizing by the total number of photon pairs recorded. We note that this experiment is not performed in a loophole-free way, but nonetheless provides us with the quantum correlations we wish to analyse.

Typically, Bell experiments aim to violate the CHSH inequality \cite{clauser_proposed_1969}, confirming that quantum mechanical systems cannot be described with local hidden variable models. 
The CHSH inequality is calculated as
\begin{subequations}
\begin{align}
    |S| & \leq 2, \text{ where } \\
    S & = \mathbb E[x_1y_1]-\mathbb E[x_2y_1]+\mathbb E[x_1y_2]+\mathbb E[x_2y_2],
\end{align}
\label{eq:S}
\end{subequations}
where $\mathbb E[xy]$ defines the correlation between Alice's ($x=\{x_1,x_2\}$) and Bob's ($y=\{y_1,y_2\}$) measurements, given as
\begin{equation}
 \mathbb E[xy] = \frac{N_{H_AH_B}^{xy}-N_{H_AV_B}^{xy}-N_{V_AH_B}^{xy}+N_{V_AV_B}^{xy}}{\mathcal N\tau}.
\end{equation}
While the CHSH inequality holds for systems which respect local causality, a pair of quantum entangled particles can achieve a maximum value of $|S| = 2\sqrt{2}$.
By tuning the $\gamma$ parameter in \autoref{eq:gamma}, then for measurement settings fixed to be optimal for the case $\gamma = 1$, we can prepare states that obey the CHSH inequality when $\gamma \leq \sqrt{2}-1$ and states that violate it.
In order to achieve the maximum violation of the CHSH inequality, it is necessary to chose specific wave plate angles for $x$ and $y$.
Here, we are not interested in violating the CHSH inequality; however, we can use it to benchmark our results from the MOEA.

\section{Edge of reality}\label{sec:num_res}

\begin{figure}[t!]
\centering
\includegraphics[width=1.0\columnwidth]{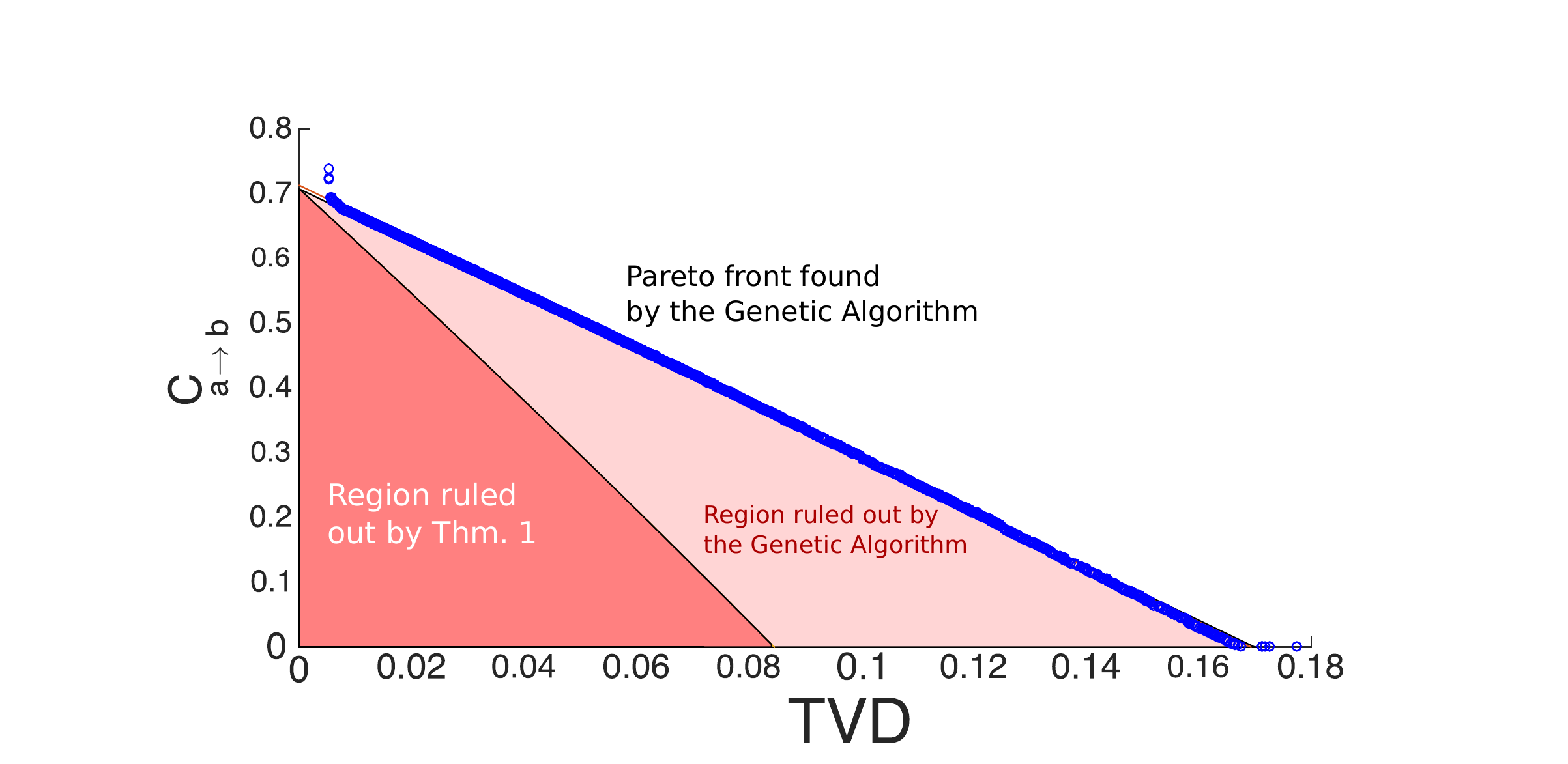}
\caption{The Pareto front for the causal network in \autoref{fig:bellgraphC} using the data from a photonic Bell experiment. The vertical axis labels the causal influence \autoref{eq:causal-influence-defn} while the horizontal axis labels the closeness to experimental data \autoref{eq:TVD}. The blue circles are the values for the non-dominated models found by the evolutionary algorithm. For comparison purposes, the straight line is a linear fit to these data.}
\label{fig:2D}
\end{figure}

\begin{figure}[ht]
\centering
\includegraphics[width=1.0\columnwidth]{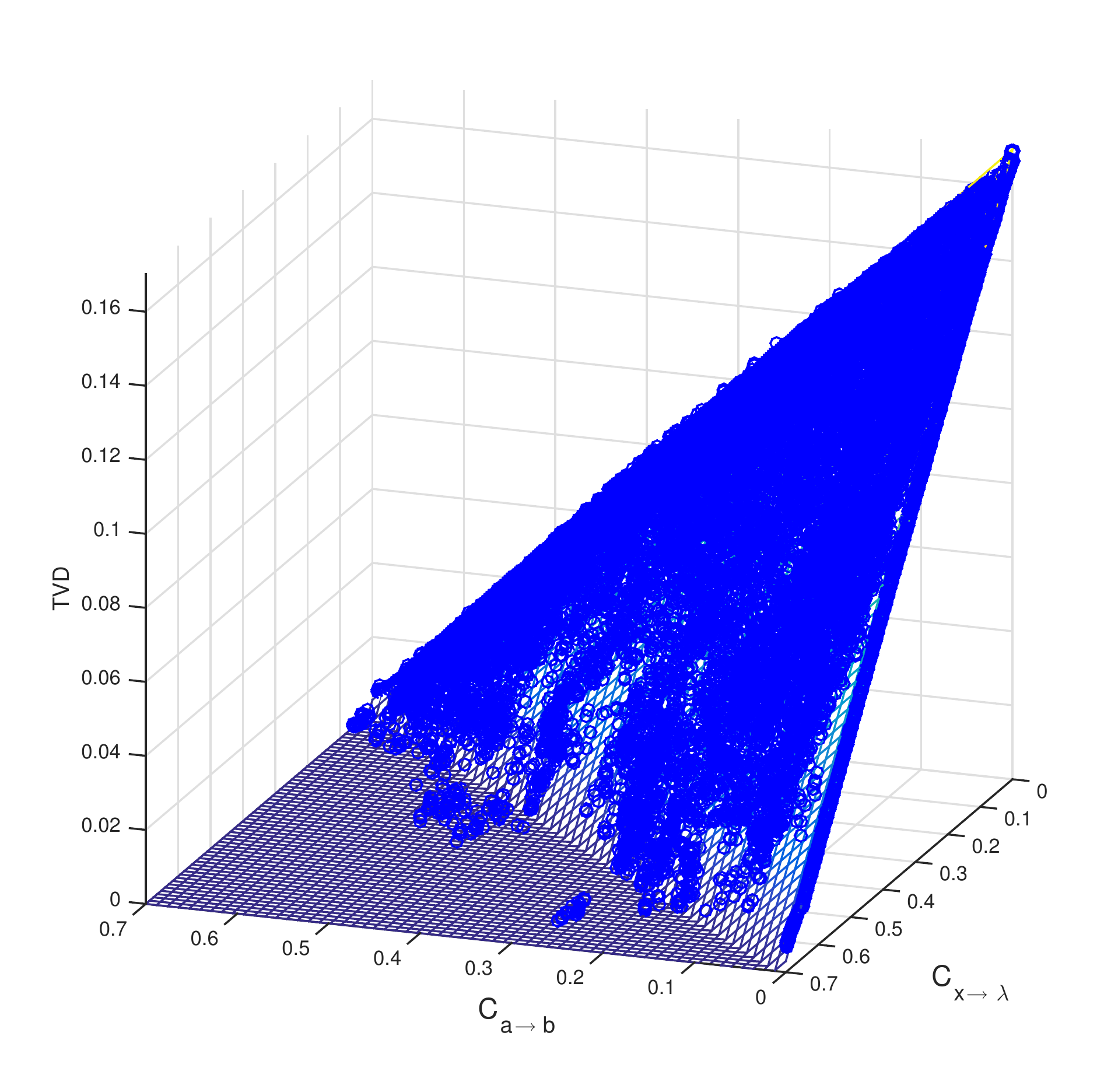}
\caption{The Pareto front for a local causal network with added $a\to b$ and $x\to \lambda$ edges using the data from a photonic Bell experiment. The vertical axis labels the closeness to experimental data \autoref{eq:TVD}. The two horizontal axes label the causal influences \autoref{eq:causal-influence-defn} for the added edges. The blue circles are the values for the non-dominated models found by the evolutionary algorithm. The flat surface is a linear fit to these data.}
\label{fig:3D}
\end{figure}

Using the experimental data (where $\gamma=0.984$), we searched for the Pareto optimal models by developing a multi-objective evolutionary algorithm to find the best underlying probability distributions for a causal network. Since this represents a trade-off between a local realistic model and real-world correlations, we call the Pareto optimal surface the ``edge'' of reality.

An \emph{individual} of the population is a probability distribution over the nodes of a given causal network (each such individual is a causal model, $M$) and its \emph{multi-objective fitness} depends on how close the model can reproduce the experimental data and the amount of causal influences between nonlocally separated variables.     

As an initial step, we examined relaxing one casual edge at a time, beginning with a causal influence from $a$ to $b$---that is, Alice's outcome is allowed to influence Bob's. The Pareto front (the numerical approximation to the Pareto optimal) is shown below in \autoref{fig:2D}. Like the theoretical bounds (which while not linear, are nearly so in the considered domain), the front appears to be linear (Pearson's $\rho^2$ value of 0.997, with bisquare robust fitting). That is, increasing locality violations allows observed (quantum) correlations to be more exactly matched, the trade-off being approximately \textit{linear} in nature.  Next, we relax the causal edges $a\to b$ and $x\to \lambda$ simultaneously. The found Pareto front is shown in \autoref{fig:3D}. Again, we see that the front appears linear ($\rho^2 = 0.9902$). We have also used our algorithm to test other causal networks and found the results to be quantitatively identical to these two cases.

\section{The Evolutionary Algorithm}\label{sec:ea}

In order to find the Pareto front of solutions, it is necessary to find feasible probability distributions that give rise to the required TVD with the required causal-violations(s). There is no known way of doing this analytically. Even in simple single edge causal models the search space is prohibitively large and objective non-convex. This search space grows rapidly with additional causal edges. Evolutionary Algorithms are known methods for finding such Pareto fronts where there is only limited knowledge of the underlying search landscape. We wish to numerically find the Pareto optimal set of models representing Bell experiment data.  To do so, we use evolutionary computation \cite{Holland75a}.

Such algorithms are generally well studied for functions of the form $f:\mathbb R^m\to \mathbb R^n$.  However, here the domain of our objective function $f$ is $\mathcal M$, i.e. the probability distributions on the causal network. 
Consequently, there are implicit constraints on the relative values these distribution can take (for instance, in each node they need to sum to 1) and so we have devised a set of evolutionary operators that allow the probability distribution of an arbitrary causal network to be evolved. With this we combine several evolutionary computation strategies to evolve and explore the Pareto front of a given arbitrary network.

\subsection{Evolutionary Algorithm Overview}
\label{sec:ea-overview}

\begin{figure}
\begin{algorithm}[H]
    \caption{Evolutionary Algorithm}
    \label{alg:mu-plus-lambda}
    \begin{algorithmic}
        \Require Population sizes $\mu,\lambda \in \mathbb{N}$,
            crossover and mutation probabilities $p_{\times},p_\mu$
        \Require an initial population $P_0$, a number of generations
            $N_\gen$.
        \Require A genetic operator \textsc{Evaluate}$(I)$
            that annotates individuals with their fitness $\vec{f}(I)$.
        \Require A genetic operator \textsc{Mutate}$(I)$ that
            mutates an individual in-place.
        \Require A genetic operator \textsc{Crossover}$(I_1, I_2)$
            that crosses over two individuals in-place.
        \Require A genetic operator \textsc{Select}$(P, \mu)$
            that selects $\mu$ individuals from the population
            $P$.
        \Ensure A Pareto front $P_*$ of individuals with respect to the
            fitness functions implemented by \textsc{Evaluate}.

        \LineComment{
            In this Algorithm, we follow the DEAP \cite{DEAP_JMLR2012} convention of storing an individuals' fitness as \emph{metadata}.
            This prevents having to re-evaluate fitnesses for every comparison.
        }

        \State $P \gets P_0$
        \State $P_* \gets \textsc{KDTree}(\{\})$
            \Comment{Initialize the Pareto front to an empty k-d tree \cite{bentley_multidimensional_1975}.}
        \State \textsc{Evaluate}($P$)
            \Comment{Evaluate each individual in the initial population.}
        \For{$i_\gen \gets 1, \dots, N_\gen$}
            \State $P_\offspring \gets \{\}$
            \While{$|P_\offspring| < \lambda$}
                \State Draw two individuals uniformly at random from $P$
                    and copy them as $I_1$ and $I_2$.
                \Switch{$u \sim \Uni(0, 1)$}
                    \Case{$u \in [0, p_\times)$}
                        \State \textsc{Crossover}($I_1$, $I_2$)
                    \Case{$u \in [p_\times, p_\times + p_\mu)$}
                        \State \textsc{Mutate}($I_1$)
                    \Case{$u \in [p_\times + p_\mu, 1]$}
                        \LineComment{Leave $I_1$ and $I_2$ unmodified.}
                \EndSwitch
                \State $P_\offspring \gets P_\offspring \cup \{I_1\}$
            \EndWhile
            \State \textsc{Evaluate}($P_\offspring$)
            \State $P \gets $\textsc{Select}($P \cup P_\offspring$, $\mu$)
                \Comment{Using the NSGA-II crowding operator, order the individuals and select the next generation from this one and the new offspring.}
            \For{$I \in P$}
                \If{there does not exist $I' \in P_*$ such that $I' \succeq I$}
                    \Comment{
                        Average time complexity $\mathrm{O}(\log|P_*|)$  for k-d trees.
                    }
                    \State $P_* \gets P_* \cup \{I\}$
                \EndIf
            \EndFor
            \If{any individuals were added to $P_*$ this generation}
                \State $P_* \gets \left\{
                    I | I \in P_* \text{ such that } \forall I' \in P_*,
                    I' \not\preceq I
                \right\}$
                    \Comment{Remove dominated individuals from the Pareto front.}
                \State Rebalance $P_*$.
            \EndIf
        \EndFor

        \State \Return $P_*$

    \end{algorithmic}
\end{algorithm}
\end{figure}

As the cornerstone of our multiobjective evolutionary algorithm (MOEA) we utilize the well-known and
well-understood NSGA-II algorithm \cite{NSGA}. Although the NSGA-II algorithm specifies both the generation and selection procedures, we utilize the the DEAP software library \cite{DEAP_JMLR2012}
 which provides the NSGA-II algorithm only for the ``select'' stage. 
 The method by which we proceed is to use the $(\mu + \lambda)$ algorithm (detailed in \autoref{alg:mu-plus-lambda}) where we set $\lambda=\mu$ to be the population size. For the purposes of avoiding confusion we note that the $(\mu + \lambda)$  is more properly an algorithm used with a subset of evolutionary algorithms known as evolutionary strategies, and thus is not part of the toolkit of the seperate branch known as genetic algorithms. Consequently, our algorithm is not strictly a genetic algorithm but is an evolutionary algorithm. Although we use  an implementation of $(\mu + \lambda)$, by setting $\lambda=\mu$ the algorithm is functionally equivalent to the generation algorithms used in genetic algorithms. In this paper we make no distinction between genetic algorithms and the more general term evolutionary algorithm in the classification of the algorithms used.
 The overall implementation of the algorithm is thus functionally identical to the original NSGA-II algorithm, save that the selection of parents is random rather than by binary tournament selection.
 
Consequently this evolutionary algorithm proceeds in
\emph{generations}, each of which consists of producing $\lambda$ offspring
from the previous generation's population, then selecting $\mu$ individuals
from the combination of the previous population and the new offspring to form
the new population. As detailed in \autoref{alg:mu-plus-lambda} the $(\mu + \lambda)$ algorithm is expressed
abstractly in terms of \emph{genetic operators} that create, crossover,
evaluate and select individuals within each population. Thus, we form our
algorithm by specifying what an individual is, the fitness functions that we
use in evaluating individuals, and by providing suitable genetic operators to create ``children''
causal networks.

\subsection{Representation of Individuals}
\label{sec:ga-individuals}

Effectively, our genetic algorithm searches for Pareto optimal models $M \in
\mathcal M$ by representing $M$ as an assignment of conditional distributions
to each node in a causal network with a fixed structure. Since the random variables at each
node are constrained to be discrete, we represent the conditional
distributions by \emph{tensors}, such that finding arbitrary joint, marginal
and conditional distributions over subsets of the nodes is then an exercise in
standard tensor contractions.

In particular, consider a node $x_i$ with $n$ causal parents $\pa_i = \{x_{i_1}, x_{i_2}, \dots,
x_{i_n}\}$.  Then, the distribution $\Pr(x_i | \pa_i) =
\Pr(x_i | x_{i_1}, \dots, x_{i_n})$ is given by the tensor
\begin{equation}
    X[j_0, j_1, \dots, j_n] \defeq \Pr(x_i = j_0 | x_{i_1} = j_1, \dots, j_n),
\end{equation}
where we have used square brackets to indicate indices 
(similar to C- or Python-style notation).

We can contract repeated indices
of two such tensors with the tensor at a corresponding third node
to perform expectation values. For example, let $A$ be the tensor for
$\Pr(a | x, \lambda)$, $B$ be the tensor for $\Pr(b|y, \lambda)$ and
$\Lambda$ be the tensor for $\Pr(\lambda)$ in the model of \autoref{fig:bellgraphC}.
Then, to find $\Pr(a, b | x, y)$, we compute
\begin{equation}
    \Pr(a, b | x, y) = \sum_{\lambda } A[a, x, \lambda] B[b, y, \lambda] \Lambda[\lambda].
\end{equation}
The general case, allowing for arbitrary numbers of random variables
and conditions, is given as \autoref{alg:ind-joints}.

\begin{figure}
\begin{algorithm}[H]
    \caption{Joint and Conditional Distribution Tensors from Individuals}
    \label{alg:ind-joints}
    \begin{algorithmic}

        \Require Individual $I$, random variables $x_1, \dots, x_n$, random variables $y_1, \dots, y_m$.
        \Ensure Tensor $J[i_1, \dots, i_n, j_i, \dots, j_m] = \Pr(x_1 = i_1, \dots, x_n = i_n | y_1 = j_1, \dots, y_m = j_m)$ for the distribution
        represented by $I$.

        \State $X' \gets \{x_1, \dots, x_n\}$
            \Comment{$X'$ holds those rvs we must still include.}
        \State $F \gets \{\}$
            \Comment{
              $F$ holds those tensor factors we include
              in the final contraction.
            }
        \While{$X'$ is not empty}
            \State $F \gets F \cup X'$
            \State $X' \gets \bigcup_{x\in X'} \pa_{x} \setminus F$
                \Comment{Add in any parents that we have not already added.}
        \EndWhile

        \State $J \gets$ Einstein sum over of $F$, holding indices
          $\{x_1, \dots, x_n, y_1, \dots, y_m\}$.
        \Comment{Marginalize over parents not appearing as $x$ or $y$.}
        \State \Return $J$

    \end{algorithmic}
\end{algorithm}
\end{figure}
\subsection{Fitness Functions}
\label{sec:ga-fitnesses}

Our algorithm uses two different kinds of fitness functions:
\begin{enumerate}
    \item The total
      variational distance (TVD) between the joint distribution computed from an individual and the observed frequencies.
    \item Causal inflences along penalized edges, as generalized from
        the definition given by \autoref{eq:causal-influence-defn} in \autoref{sec:bell}.
\end{enumerate}

Dealing with each in turn, the TVD is calculated by taking the vector
1-norm between the flattened joint distribution tensor
the observable variables calculated as in
\autoref{sec:ga-individuals} and the flattened observed frequencies,
\begin{equation}
  f_{\TVD}(I) = \|J_\obs^\flat(I) - F^\flat\|_1,
\end{equation}
where $I$ is an individual with joint distribution tensor $J_\obs(I)$
over all observables, $F$ is
the tensor of observed frequencies, and where $\flat$ indicates flattening---
that is, reduction of an arbitrary-rank tensor to a rank-1 tensor.

As discussed in \autoref{sec:bell}, we adopt a definition of causal influence
that allows us to reason even in lieu of interventions. Our definition of the
causal influence $C_{a\to b}(I)$ for an individual $I$ proceeds in three steps.
First, we maximize over pairs of settings of $a$ to find which are most distinguishable
through observations of $b$ alone. We then maximize over the conditions under which
these observations are made, represented by maximizing over feasible assignments to
the parents of $b$. Finally, we \emph{marginalize} over those nodes which are also
parents of $a$ to prevent ``hiding'' causal influence; this is illustrated in
\autoref{fig:grandparent-measure-example}.

\begin{figure}
    \begin{centering}
           \includegraphics{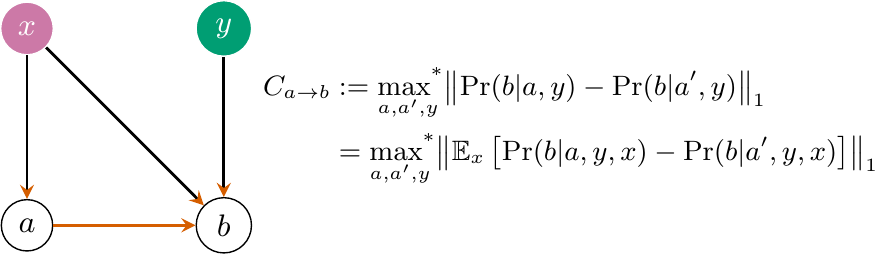}
   


    \end{centering}
    \caption{
        \label{fig:grandparent-measure-example}
        An example of the causal influence measure $C_{a\to b}(I)$
        given by \autoref{eq:causal-influence-defn} applied to a
        more complicated graph. The random variable $y$ is conditioned on and maximized
        over, as it is a parent of $b$ but not of $a$. By contrast, the variable
        $x$ is also a parent of $a$ and so it is marginalized over,
        resulting in the causal influence definition at right.
    }
\end{figure}

\subsection{Genetic Operators}
\label{sec:ga-operators}

Having defined the mapping from models to individuals, we complete the
specification of our algorithm by detailing the various genetic operators
which act on these individuals.

\paragraph{Creation} \label{sec:ga-creation} In order to create a new individual $I$, we must specify
a new conditional distribution at each node of the causal graph. We do so
randomly by assigning a tensor $X_i(I)$ with entries drawn uniformly from $[0, 1]$
to each node, then renormalizing to ensure\\ $\sum\limits_{ x_i} \Pr(x_i | \pa_i)=1$.
Using the tensor notation defined above,
\begin{align}
    \label{eq:create-renormalize}
    X_i(I)[j_0, j_1, \dots, j_{|\pa_i|}] = \frac{
        \tilde{X}_i(I)[j_0, j_1, \dots, j_{|\pa_i|}]
    }{
        \sum_{j'_0} \tilde{X}_i(I)[j'_0, j_1, \dots, j_{|\pa_i|}]
    },
\end{align}
where $\tilde{X}_i(I)[j_0, \dots, j_{|\pa_i|}] \sim \Uni(0, 1)$ is the unnormalized
tensor of $I$ at $X_i$.

\paragraph{Crossover/Mating} \label{sec:ga-crossover}Given two individuals $I_1$ and $I_2$, we
mate them to produce two new individuals $I_1'$ and $I_2'$ by swapping
the tensors at each node with probability $p_\chi = 0.5$. That is, for each
node $x_i$, the corresponding tensor $X_i(I_1')$ of $I_1'$ is given by
\begin{equation}
  X_i(I_1') = \begin{cases}
      X_i(I_2) & \text{with probability } p_\chi \\
      X_i(I_1) & \text{with probability } 1 - p_\chi
  \end{cases}.
\end{equation}

\paragraph{Mutation} \label{sec:ga-mutation} Given a single individual that has been selected for
mutation, we proceed by first picking a node $x_\mu$ on the causal graph
uniformly at random, with corresponding tensor (assuming $n$ parents)
$X_\mu[i_0, \dots, i_n]$. 

One of the conditional events represented by $X_\mu$ (that is, a single element of the tensor) is selected at random and the value (and hence the probability assigned to the selected outcome) is randomly increased or decreased by a sample from a zero-mean Gaussian distribution, where the variance is a user supplied parameter. The mutated element of $X_\mu$ is then clipped to the interval $[0, 1]$, and the relevant tensor index renormalized
such that \autoref{eq:create-renormalize} holds.

By way of example, if we had a node $a$ with binary values, which in turn had one parent $x$ also with binary values, then the information pertaining to that node would be stored in a $2\times 2$ tensor $A[a, x]$ corresponding to the probability distribution
\begin{align}
\Pr(A = a | X = x) & = A[a, x] =
\kbordermatrix{
&x=0&x=1\\
a=0& \alpha & \beta\\
a= 1 & \gamma & \delta \\
},
\end{align}
where $\alpha$ represents the probability of $a$ being 0 given $x$ is 0 and so on. From this it can be seen that it is necessary that $\alpha + \gamma = \beta + \delta = 1$. One of $\alpha$, $\beta$, $\gamma$ or $\delta$ would be modified as discussed above, and the remaining values renormalized to ensure that the relevant probabilities continue to sum to 1.

\paragraph{Selection} \label{sec:ga-selection} For selection we used an unmodified version of the NSGA-II algorithm \cite{NSGA}. NSGA-II uses a fast sort algorithm to locate the non-dominated individuals and then applies a crowding distance sorting algorithm to prefer those individuals that explore different parts of the pareto front.
The ``best'' $\mu$ individuals are retained for the next generation.

\subsection{Decomposition of the multi-objective optimization (the Island Model)}\label{sec:island}

\begin{figure*}[ht]
    \begin{centering}
        \includegraphics[width=1.5\columnwidth]{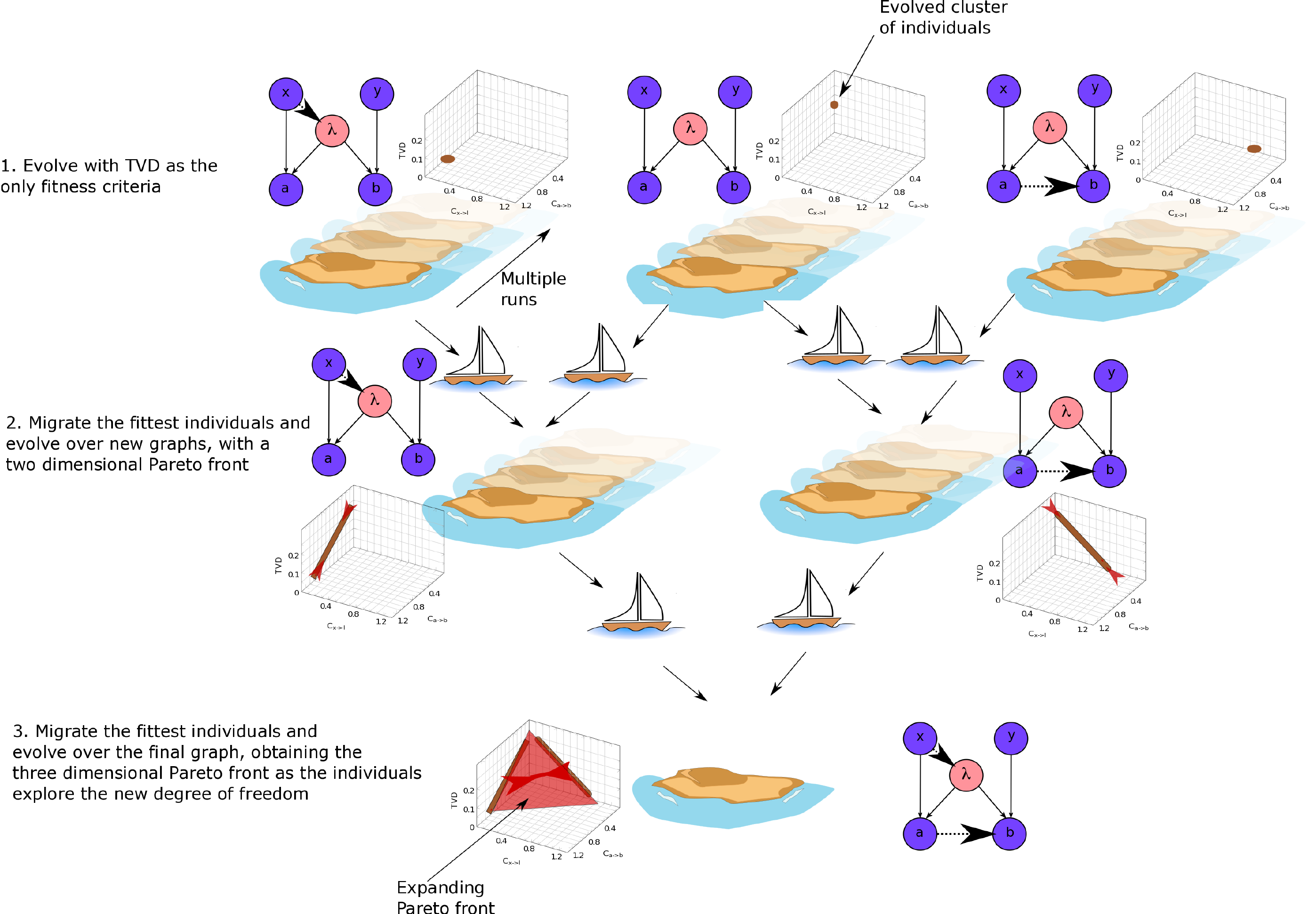}
     \caption{ \label{fig:islands}
        \emph{Island model} diagram showing the steps in evolving a three-dimensional Pareto front. This allows a the edges of the Pareto front to be found by exploring lower dimensional graphs with lower populations. Multiple runs in each of step 1 and step 2 can be done concurrently. 
    } \end{centering}
\end{figure*}

Here we present an enhancement to the basic genetic algorithm discussed above that aids the discovery of the global Pareto front in multi-dimensional scenarios, where---as is the case here---it is possible to evolve populations to occupy the extremes of any particular front.

As discussed in \autoref{sec:ea_implement} it is well known that the NSGA-II crowding becomes less effective with the exponential increase in the size of the front with the number of dimensions. However, in our case are we able  to force the population to start at extreme points of the Pareto front by pre-evolving the population on structurally reduced graphs or with reduced fitness criteria. These populations are able to seed the graph we wish to explore and spread over the front, fleshing it out over multiple runs. This can then be repeated as we increase the dimensions of the fronts. This is not dissimilar to the mechanism used in NeuroEvolution of Augmenting Topologies (NEAT) \cite{NEAT} where populations are evolved on small neural networks prior to allowing additional links to be added. A similar idea of decomposing the objectives is explored in \cite{liu2014}. Effectively, where a multiple dimension Pareto front needs to be explored, different populations are evolved on all permutations of the simpler graphs (on seperate ``islands'') before being brought together for evolution over the full graph. This is illustrated in \autoref{fig:islands}.

This technique allows us to find a three-dimensional Pareto front based on a graph with two causal edges, from $x\rightarrow \lambda$ as well as $a\rightarrow b$. This was evolved using five runs of the ``island model'' detailed above. For each run the initial islands had a population of 300 and comprised 4 runs of 400 generations. The initial runs found populations clustered around the three extremes: (1) min(TVD), hold $C_{x\rightarrow \lambda}=0$, $C_{a\rightarrow b}=0$; (2)
 min($C_{a\rightarrow b}$), hold TVD $= 0$, $C_{x\rightarrow \lambda}=0$; and (3) min($C_{x\rightarrow \lambda})$, hold TVD $= 0$, $C_{a\rightarrow b}=0$. 

The second set of islands take the relevant individuals generated above, reduce them to the best 400 individuals representing the extreme of the Pareto fronts for that island and transplant them to expanded causal network graphs. In this case there are two second generation islands: one generating the two-dimensional Pareto front for \{TVD and $C_{a\rightarrow b}$\}, with $C_{x\rightarrow \lambda}$ held to be 0 (i.e. no causal $x\rightarrow \lambda$ link); and the second the two-dimensional Pareto front \{TVD and $C_{x\rightarrow \lambda}$\}, with $C_{a\rightarrow b}$ held to be 0. These populations are then evolved on the respective causal networks generating two dimensional Pareto Fronts similar to \autoref{fig:sfigABTVD}.  These populations are placed in an $\epsilon$-Dominance archive. (In other-words they are only kept if they dominate all previous individuals by at least $\epsilon$, where in this implementation $\epsilon$ was $10^{-8}+10^{-5}*|value|$). The entire process so far is repeated  several times (in this experiment 5 times)  to ensure we have 2,000 suitable individuals in the archive. These individuals are, effectively, clustered on the two dimensional fronts specified in the second set of islands. This final population is used to generate the 3-dimensional Pareto front shown in \autoref{fig:approxFront}. The final island had a population of  2,000 individuals (extracted initially from the $\epsilon$-Dominance archive), evolved for 800 generations. This constituted one run of the island model. The Model was run 5 times, with every individual generated by the model being submitted to (but not necessarily accepted by) the global $\epsilon$-Dominance archive. 

To illustrate the advantage of using this model, we have also plotted (in red) the best Pareto front found using just the basic algorithm (i.e. evolving only over the full graph). These additional points were collected over 8 runs, using a high population (6,000) and represented five times the computing power required for the Island Model. As can be seen the global  $\epsilon$-Dominance archive for  the basic algorithm contained few individuals on the best  global Pareto front found by the Island Model. The Pareto front for the Island Model (plotted blue) does appear to be a viable candidate for the actual global front, indicating that - for this model - the trade-offs in the different causality violations considered is linear.  The front fits a linear plane with a Pearson's $\rho^2$ value of 0.9902, the non-fitting points being those with extremely low TVD values (the points which appear on the horizontal part of the mesh). While this still needs to be investigated further we believe it is related to experimental noise which might require increased causal violation to match the noisy data exactly. So far as we are aware this is the first evidence that tradeoffs in multiple causal violations are also linear for such graphs.
\vspace{12pt}

\begin{figure}
\includegraphics[width=0.8\linewidth]{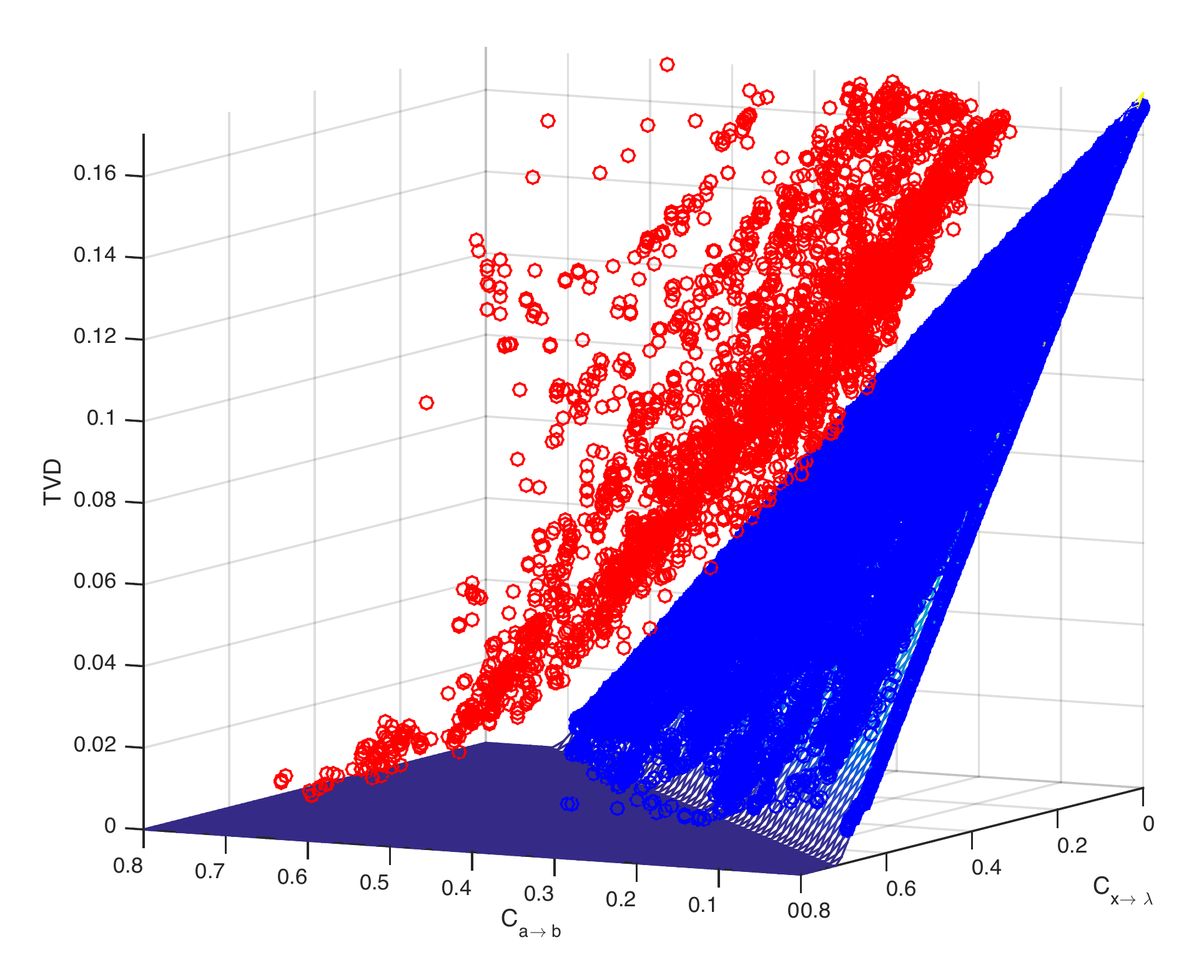}
\caption{The 3D Pareto front for $a\rightarrow b$ and $x \rightarrow \lambda$ violations found using the ``Island Model'' is shown in blue. Beneath it is a linear mesh which fits the Pareto front with a Pearson's $\rho^2$ of 0.9902. By Comparison the non-optimal Pareto front found by combining the best individuals from 8 runs using the basic algorithm (i.e. no pre-evolution) on the full graph is shown in red. The front shown in blue took a fraction of the computing time to find compared to the non-optimal red runs.}
\label{fig:approxFront}
\end{figure}

\subsection{Previous Work and Design Decisions}\label{sec:previous}

Although there has been previous work in using genetic algorithms to explore Bayesian causal networks (e.g. \cite{Guimaraes2013}, \cite{Muruzabal2007} and \cite{Larranaga2013}), the focus of such works has been to create the network and the links therein. For instance, \cite{Ross2006} uses a multi-objective genetic algorithm (MOGA) to evolve dynamic Bayesian networks. There the multi-objectives explored were the ability of the evolved networks to explain the data, compared with the complexity of the network in question. The MOGA was used instead of, for instance, a minimum description length (MDL)  constraint. In all of these cases the network is being used to model something of interest and then, given some observed values, infer the likely causes. The genetic algorithms are used to construct different models which are then trained, typically the success (or otherwise) of a particular model being its performance on withheld data.

Our work differs because of the way we wish to utilize the Bayesian causal networks, specifically we specify the networks we are interested in, namely those which model a physical view of ``reality'' with specified local causality violation. 
Training such networks to replicate observed correlations is of limited interest because successful training results in one specific probability distribution that explains the data. What we are interested in  finding are all the relevant probability distributions where the ability to match the observed correlations is contrasted with the strength of the local causality violations. The evolutionary algorithm is used, not to evolve networks, but rather to find these probability distributions given the network. The MOEA is used to guide evolution along these Pareto fronts.

In order to explore the Pareto front some type of MOEA algorithm is required. MOEA on two or three dimensions are relatively well understood. Algorithms to explore large dimensions are still an active area of research (see for example \cite{NSGAIII}). Since our initial experiments (reported here) would only require causal networks with no penalized edges (a single-value optimization), one penalized edge (a MOEA with a two-dimensional Pareto front) or two penalized edges (a MOEA with a three-dimensional Pareto front), we decided to use the well understood NSGA-II \cite{NSGA}. Although NSGA-II attempts to return the whole of the Pareto front in a single run it was quite clear that the search space (being the required probability assignments for all the nodes in the network) was not smooth, even though the Pareto fronts may be (and, in fact, turned out to be) smooth. Given this an $\epsilon$-Dominance Archive \cite{laumanns2002} was maintained and updated through multiple runs. 
In order to maintain diversity between runs the archive was not used to guide the evolution, but rather served as an updated archive of the best Pareto front found so far.
After completing multiple runs, the individuals in the archive thus represent the Pareto front
for the entire procedure, rather than for each run taken in isolation.

In \autoref{sec:island} we describe how the $\epsilon$-Dominance Archive generated from a lower dimension front can be used to seed evolution when a higher dimensional Pareto front is explored, in a manner not dissimilar to the algorithm presented by \citet{liu2014} or utilized by \citet{NEAT}.

\subsection{Implementation methodology and details}\label{sec:ea_implement}

Our initial runs with the EA (i.e. the MOEA with a single objective, being to minimize the TVD) were used to verify that the EA could match known results. In this case we start with a causal network that reflects Bell's non-locality assumptions (as shown in \autoref{fig:bellgraphA}) and for various values of  $\gamma$ in \autoref{eq:gamma} use the EA to try and match the experimentally observed joint probability distribution. This is a single-objective EA, with fitness being governed solely by the TVD, i.e.\ by how closely the observed probability distribution of the Model (being the observed joint probability distribution for an individual over the local causal graphs) matches the experimental data. It is known that when $\gamma$ is less than $\sqrt{2}-1$ the observed data can be modeled with a local causal network. The TVD values should increase as $\gamma$ increases to 1 since the empirical distribution no longer factorizes into a locally causal distribution.  

An initial population of 300 ($\mu=300,\lambda=300$) was chosen, with the probability of crossover, being 0.1 (and mutation 0.9). The mutation operator used a standard deviation of 0.1 (see \autoref{sec:ga-operators}). As is typical for experiments using genetic algorithms no systematic attempt was made to find the ``best'' parameters for the algorithm. Rather during the course of some initial testing runs, runs with variations of parameters were tried and the parameters of the ones that seemed to find solutions quickest were used. Population sizes reflected those minimum populations required to avoid runs being trapped early on in local minima. The parameters reported are not reported in a claim of optimality, but rather are reported for the purposes of reproducibility.  In any case, once a solution is obtained, its validity does not depend on the means by which it was found.

\autoref{fig:gammagraph}a shows experimentally measured density matrices for a range of state $\gamma$ values. The reduction in coherence is observed as decreasing off-diagonal terms.  \autoref{fig:gammagraph}b shows the minimum TVD values emerging from 20 runs of the graph for various $\gamma$ values together with a linear line fitting the data, running  from the known y-intercept of 0 TVD for $\gamma=(\sqrt{2}-1)$. As can be seen the EA fits the expected linear results (Pearson's $\rho^2=0.9952$).

Having ensured that the algorithm could correctly match the known results on a  causal network consistent with Bell's non-locality assumptions, the next stage is to require a relaxation of local causality to allow the EA (now operating as a MOEA)  to match the correlations present in entangled states. As an initial step, we examined relaxing one casual edge at a time, beginning with a causal influence from $a$ to $b$---that is, Alice's outcome is allowed to influence Bob's. This is illustrated in \autoref{fig:bellgraphC}. This becomes a multi-objective problem with a two-dimensional Pareto front, ostensibly well within the capabilities of NSGA-II. The tensor contractions required are not overly complex but with increasing numbers of hidden variables (as a result of additional causal links) each run takes a non-trivial amount of time. As is typical where the search landscape (being the underlying conditional probability distributions) is not smooth (even though the fitnesses such distributions reduce to are smooth) a number of runs failed to converge to any part of the Pareto front, with most runs finding part, but not all of the Pareto front. In \autoref{fig:multixy} we show the individual results of 40 such runs. As can be seen from the figure just under half of the 40 runs had a large percentage of their front non-opitmal, with approximately half the runs being plotted on top of each other on the Pareto Front. To generate the final Pareto front each of the individuals in the  $\epsilon$-Dominance archive from each of the 40 runs are submitted to the \emph{global}  $\epsilon$-Dominance archive, so that the best estimate of the true Pareto front can emerge, as shown in \autoref{fig:multixyGlobal}.

Two points arise from these results. The first is that the front appears to be linear, that is increasing locality violations allows observed (quantum) correlations to be more exactly matched, the tradeoff being \textit{linear} in nature. As far as we are aware this was not previously known. The second arises from the number of failed or only partially successful runs. In particular we note that while in the majority of runs the MOEA was able to find many points on (or close to) the Pareto front, other runs could be trapped and all runs had difficulty at either extreme of the front.  It is clear that the search landscape in general is not smooth - the interplay between the conditioning on the hidden local variable and the other probability distributions allow the MOEA to become trapped in some local minima. The larger front (in this case a two-dimensional line) allowed the population to ``slide'' away from the edge cases. In addition as observed in \cite{Wood2015The} it is likely the edge cases represent very specific distributions. Whilst some of these observed difficulties could, in part, be ameliorated by using a larger population and relying on the NSGA-II crowding mechanism to prevent such slippage, as is known this will not be feasible if the front consists of three (or more) dimensions. The front grows exponentially with the number of dimensions, requiring an exponential increase in population size. An alternative MOEA such as NSGA-III may help but each alternative comes with their own difficulties and assumptions. It is, however, possible to use the specifics of the problem space to address these concerns. We know we can evolve the population on a more limited graph (such as the purely local graph) and force the population to find the lowest TVD with a causal violation of zero (i.e.  in the local graph $C_{a\rightarrow b}=0 \text{, since there is no link } a\rightarrow b$). This evolved population can then be ``transplanted'' on to a graph that does have an $a\rightarrow b$ link (e.g.  \autoref{fig:bellgraphC}). The other extreme (i.e. lowest  $C_{a\rightarrow b}$ violation for TVD$=0$) can be found with a small alteration to the fitness function. To find this point we evolve the population on the graph representing \autoref{fig:bellgraphC} but with a single-objective fitness function, implemented as minimizing the TVD, but where two individuals have the same TVD, the one  with the lowest causality violation is preferred. This drives the population towards zero TVD and then minimizes the causality. Even with this the observed correlations were unable to achieve an exact $\rm{TVD}=0$, it is speculated this is a result of experimental noise. The ability to generate populations (distributions) that sat at the extreme points of the Pareto front allowed the entire two-dimensional front to be revealed and, as discussed below, can similarly be utilized to reveal the three-dimensional front created by two simultaneously relaxed local causality constraints.

\begin{figure*}[ht]
\scalebox{.75}{
    \subfloat[]
    {
        \label{fig:multixy}
        \includegraphics[width=.6\linewidth]{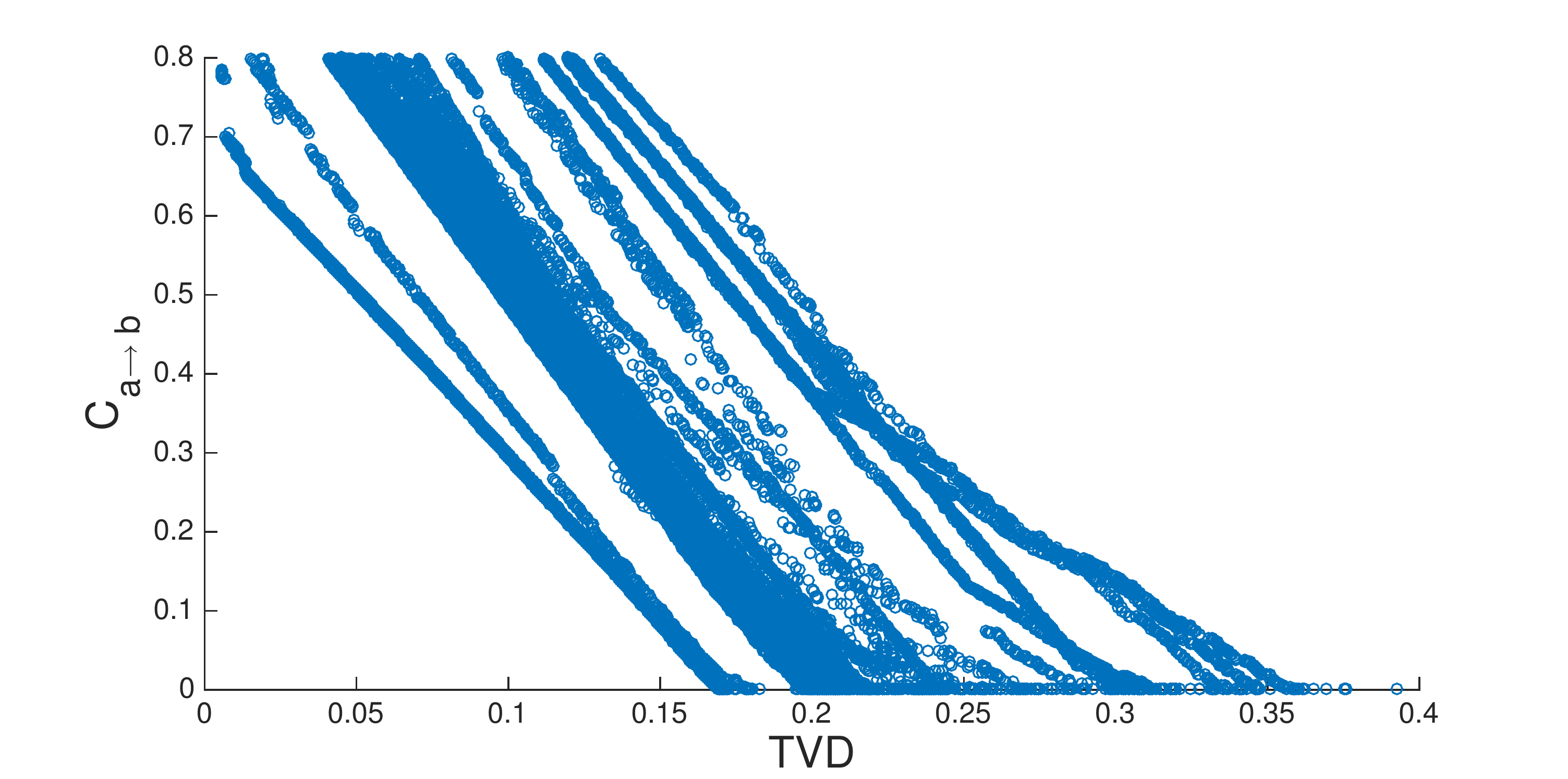}
    }
 \subfloat[]{
         \includegraphics{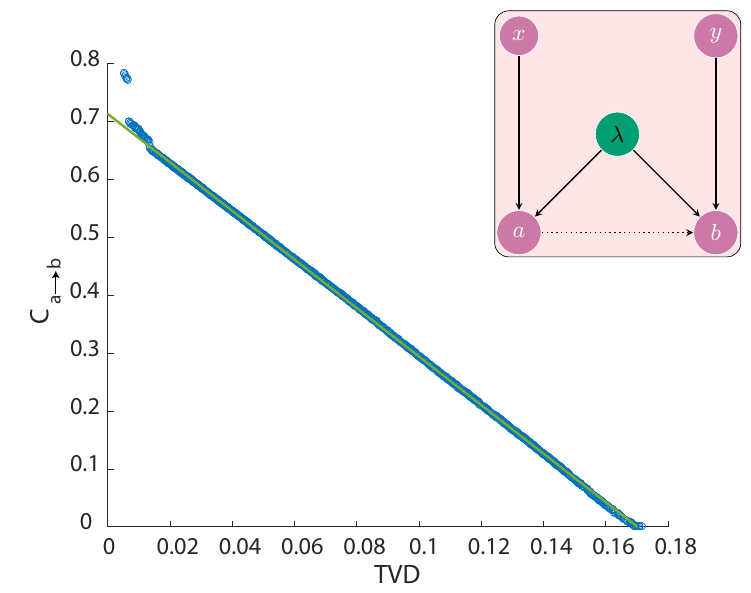}
   
	\label{fig:multixyGlobal}
  
    }
    }
    \caption{ \label{fig:sfigABTVD} The results of 40 typical runs of the EA with a causal graph allowing ``Alice to Bob'' local causality violations (see \autoref{fig:bellgraphC}). (a) For the purpose of producing this graph, each run had its own $\epsilon$-Dominance archive. As can be seen several runs failed to find the correct front at all, indicating that the interplay between hidden nodes and conditioned variables results in a non-trivial search when attempting to match the observed distributions of observable variables. (b) At the conclusion of these 40 runs the $\epsilon$-Dominances archive are combined to form the \emph{gobal}  $\epsilon$-Dominance archive generating the best estimate of the actual Pareto. The front is well fit by a straight line (Pearson's $\rho^2$ value of 0.997, with bisquare robust fitting). Exact fitting of the distribution (very low TVD) requires additional causality violation. Experimental noise might be reason for this. Although around 100 runs were conducted to produce the reported results (\autoref{fig:2D}) very few additional points were found on the Pareto Front.
    }
\end{figure*}

\section{Discussion}\label{sec:discussion}

In this work, we have developed a method to allow the study of several relaxations of local hidden variables models simultaneously in a single framework using the tools of casual networks and genetic algorithms.  

With further refinement, we hope that our approach can shed light on other scenarios where quantum correlations display richer structure than classical systems would allow.  For example, generalizations of the standard Bell scenario to more stations \cite{Werner2001All,Zukowski2002Do} and more outcomes \cite{Collins2002Bell,Collins2004A}, as well as multiple hidden variables \cite{Branciard2010Characterizing, Fritz2012Beyond,Henson2014Theory}.  In the latter scenario, very little is known since classical correlations are no longer given by linear constraints.  Very recently, Chaves has used the framework of causal networks to systematically study such higher-order constraints \cite{Chaves2015Polynomial}.  Such measured quantities will be particularly useful to our approach as they can be seen as highly relevant coarse grainings of the exponentially growing data space.  Such dimension reduction techniques will be crucial for scaling up our numerical algorithm to the analysis of multi-party quantum correlations.

In addition, there is nothing specifically ``quantum'' about our core numerical methods.  Thus, our approach should find application outside of the problem of understanding quantum correlations.  Recently Lee and Spekkens have also used inspiration from the causal analysis of quantum correlations to develop new causal discovery protocols \cite{Lee2015Causal}.  Like Lee and Spekkens, we depart from the usual considerations of observed correlation to considering the entire joint probability distribution.  Our goals differ, however; whereas the aim of Lee and Spekkens 
is to find all casual models consistent with data, our goal is to find non-dominated models  of the plausible correlations.  These two approaches are likely to find a harmonious union in the future.

\begin{acknowledgments}
We thank Martin Ringbauer and Andrew White for discussions. This work was supported by the US Army Research Office grant numbers W911NF-14-1-0098 and W911NF-14-1-0103, and by the Australian Research Council Centre of Excellence for Engineered Quantum Systems CE110001013. 
STF acknowledges support from an Australian Research Council Future Fellowship FT130101744. AP acknowledges an Australian Research Council Discovery Early Career Researcher Award, Project No. DE140101700 and an RMIT University Vice-Chancellor’s Senior Research Fellowship.
RK acknowledges additional support from the Excellence Initiative of the German Federal and State Governments (Grants ZUK 43 \& 81) and the DFG.
\end{acknowledgments}

\bibliography{causal-discovery}

\appendix
\onecolumngrid

\section{Derivation of the inequality presented in \autoref{sec:bounds}}
\label{app:bounds_derivation}

For the sake of being self-contained, let us start this section with reviewing some basic facts about discrete probability distributions and introduce some notation.
Throughout this section, we focus on the empirical frequencies $F(a,b,x,y)$ and the probability distribution $\Pr(a,b,x,y|M)$ associated to a fixed model $M$.
Here $a,b$ are the special nodes whose causal relationship is of interest, $y$ will denote the parents that are not grandparents of $b$, and $x$ is any set of additional random variables which might include hidden variables $\lambda$ as well as additional measurement outcomes. 
Therefore, the discussion is completely general and not specific to the models considered, e.g., in the experiment.

If we marginalize these distributions over any variable, say $y$, we produce new distributions
\begin{equation}
F(a,b,x) = \sum_y F(a,b,x,y)
\quad \textrm{and} \quad
    \Pr(a,b,x|M) = \sum_y \Pr (a,b,x,y|M),	\label{eq:marginalization}
\end{equation}
respectively. As outlined in \eqref{eq:marginalization}, we indicate marginalization over any variable, by simply omitting the corresponding variable in the description. 
Having such a notation at hand, the \emph{product rule} (for discrete probability distribution) assures
that as an immediate consequence of the definition of conditional distributions,
\begin{equation}
F(a,b,x,y) = F(a,b,x|y)F(y)
\quad \textrm{and} \quad
 \Pr(a,b,x,y|M) = \Pr(a,b,x| y, M) \Pr(y|M)
\end{equation}
for the variable $y$. Analogous formulas are true for any combination of the variables present in the distributions (i.e. $\left\{a,b,x,y\right\}$ for $F(\cdot)$ and $\left\{ a,b,x,y \right\}$ for $\Pr(\cdot | M)$).

With these rules and notational concepts at hand, the following statement is an immediate consequence of the triangle inequality.

\begin{lemma} \label{lem:contraction}
Let $F(a,b,x,y)$ and $\Pr(a,b,x,y|M)$ be as above. Then
\begin{equation}
\| \Pr(b|M) - F(b) \|_1 \leq \| \Pr(a,b|M) - F(a,b) \|_1 \leq \| \Pr(a,b,x,y|M) - F(a,b,x,y) \|_1 = \TVD(M).
\end{equation}
\end{lemma}

This Lemma encapsulates two particular instance of the well-known fact that marginalization contracts the total variational distance.
Since the latter is a measure of how well two probability distributions can be distinguished and marginalization corresponds to ignoring certain variables, \autoref{lem:contraction} can be intuitively paraphrased as: ``knowing more doesn't hurt''.

\begin{proof}[Proof of \autoref{lem:contraction}]

Inserting the definitions of marginalization and total variational distance yields
\begin{align}
\| \Pr(b|M) - F(b) \|_1 
= & \sum_{b} \left| \sum_a \left( \Pr(a,b|M) - F(a,b) \right) \right| 
\leq  \sum_{a,b} \left| \Pr(a,b|M) - F(a,b) \right| 
= \left\| \Pr(a,b|M) - F(a,b) \right\|_1
\end{align}
upon employing the triangle inequality. The second inequality can be established in complete analogy.
\end{proof}

We are now ready to establish the main auxiliary result necessary to establish \autoref{thm:bounds}.
It requires the concept of the \emph{harmonic mean} for two variables. For $x_1,x_2 >0$ the harmonic mean is defined as $H(x_1,x_2) = \frac{2x_1 x_2}{x_1 + x_2}$.

\begin{lemma} \label{lem:auxiliar}
Consider two bivariate probability distributions $p(u,v)$ and $q(u,v)$ over finitely many elements labeled by $u$ and $v$, respectively. 
Then, the following inequality is valid for any fixed variable $v$:
\begin{equation}
\| p (u|v) - q (u|v) \|_1 \leq \frac{\sum_u \left|p (u,v) - q(u,v) \right| + |p (v) - q(v)|}{H \left( p(v), q(v) \right) } \label{eq:auxiliar}
\end{equation}
\end{lemma}

We point out that this estimate is responsible for introducing the on first sight unfavorable  scaling of the bounds~\eqref{eq:bound}. 
However, inequality \eqref{eq:auxiliar} is actually tight, making the aforementioned behavior essentially unavoidable. 
To see this, let $u,v, \in \left\{ 0,1 \right\}$ be binary variables and let $p$ be the uniform probability distribution over the four possible joint instances. 
If one chooses $q$ to be a perfectly correlated bivariate distribution---i.e. $q(0,0) = q(1,1) = 1/2$---it is easy to see that equality is attained in the assertion of \autoref{lem:auxiliar}. 

\begin{proof}[Proof of \autoref{lem:auxiliar}]
Fix an arbitrary label $v$. Inverting the product rule allows us to rewrite the left hand side of \eqref{eq:auxiliar} as
\begin{equation}
\| p (u|v) - q(u|v) \|_1 = \left\| \frac{ p (u,v)}{p(v)} - \frac{q(u,v)}{q(v)} \right\|_1 = \frac{1}{p(v) q(v)} \sum_u \left| q (v) p(u,v) - p(v) q(u,v) \right|. \label{eq:auxiliar_aux1}
\end{equation}
For $p(v)$ and $q(v)$ we now define
\begin{equation*}
\mu := \frac{1}{2} \left( p (v) + q(v) \right) 
\quad \textrm{and} \quad
\delta := \frac{1}{2} \left( p (v) - q(v) \right)
\end{equation*}
which obey $p(v) = \mu + \delta$ as well as $q(v) = \mu - \delta$ by construction. Inserting these decompositions into \eqref{eq:auxiliar_aux1}
reveals
\begin{align*}
p(v) q(v) \| p (u|v) - q(u|v) \|_1
=& \sum_u \left| (\mu - \delta) p (u,v) - (\mu + \delta) q(u,v) \right| \\
=& \sum_u \left| \mu \left( p (u,v) - q(u,v) \right) - \delta \left( p(u,v) + q(u,v) \right) \right| \\
\leq & \mu \sum_x \left| p(u,v) - q(u,v) \right| + | \delta | \sum_u \left( p(u,v) + q(u,v) \right) \\
=& \mu \sum_u \left| p (u,v) - q(u,v) \right| + | \delta | (p(v) + q(v) ) \\
=& \mu \left( \sum_u \left| p(u,v) - q(u,v) \right| + 2 | \delta| \right),
\end{align*}
where we have employed the triangle inequality and the definition of marginalization. 
Replacing $\mu$ and $\delta$ with the original expressions then yields
\begin{align*}
\| p (u|v) - q(u|v) \|_1
\leq & \frac{ p(v) + q(v)}{2p(v)q(v)} \left( \sum_u \left| p (u,v) - q(u,v) \right| + \left| p (v) - q(v) \right| \right).
\end{align*}
The desired statement then follows from this estimate by identifying the pre-factor as $1/H\bigl(p(v),q(v)\bigr)$.
\end{proof}

We can now show that a bound holds that relates the maximum deviation between the causal influence of any fixed model $M$ and the frequencies $F$. 

\begin{lemma} \label{lem:fixedM}
For any fixed model $M \in \mathcal{M}$ and fixed set of empirical frequencies $F$, let $y$ denote the parents that are not grandparents of the random variable $b$. Then the following inequality holds 
\begin{equation}
|C_{a \to b} (F) - C_{a \to b} (M)| \le \frac{4 \; \TVD(M)}{\minz{a,y} H\bigl( \Pr(a,y|M),F(a,y)\bigr)}\,,
\label{eq:fixedMbound}
\end{equation}
where $\minz{}$ denotes the minimization over feasible assignments to the variables $a, y$.
\end{lemma}

\begin{proof}
Choose an arbitrary model $M \in \mathcal{M}$. To ease notation, denote by $y$ the parents that are not grandparents of the variable $b$. In order to derive the upper bound presented in \autoref{eq:bound}, we start with inserting the definition \autoref{eq:causal-influence-defn} of $C_{a \to b}(M)$ and observe
\begin{align}
C_{a \to b} (M)
=& \maxz{a,a',y} \left\| \Pr(b|a,y,M) - \Pr(b|a',y,M) \right\|_1 \nonumber \\
=& \maxz{a,a',y} \left\| \Pr(b|a,y,M) - F(b|a,y) - \Pr(b|a',y,M) + F(b|a',y) + F(b|a,y) - F(b|a',y) \right\|_1 \nonumber \\
\leq & \maxz{a,a',y} \left\| \Pr(b|a,y,M) - F(b|a,y) \right\|_1 + \maxz{a,a',y} \left\| \Pr(b|a',y,M) - F(b|a',y) \right\|_1 + \maxz{a,a',y} \left\| F(b|a,y) - F(b|a',y) \right\|_1 
\nonumber \\
=& 2 \maxz{a,y} \left\| \Pr(b|a,y,M) - F(b|a,y) \right\|_1 + C_{a \to b} (F), 
\label{eq:bounds_aux1}
\end{align}
where we have identified the last term as the empirical average causal effect defined in \autoref{eq:data_causal_effect}. 
As a simple bookkeeping device, let us define $v = (a,y)$ to be the cartesian product of the random variables $a$ and $y$. The first term in \eqref{eq:bounds_aux1} can be bounded by invoking \autoref{lem:auxiliar} and \autoref{lem:contraction}. 
Doing so results in 
\begin{align}
C_{a \to b} (M)-C_{a \to b}(F) 
\leq & 2 \maxz{v} \left\| \Pr(b|v,M) - F(b|v) \right\|_1\\
\leq & 2 \maxz{v} \frac{\sum_b \left| \Pr(v,b|M) - F(v,b) \right| + \left| \Pr(v|M) - F(v) \right| }{H\bigl( \Pr(v|M), F(v)\bigr)}\\
\leq & 4 \maxz{v} \frac{ \| \Pr(v,b|M) - F(v,b) \|_1}{H\bigl( \Pr(v|M), F(v)\bigr)}\\
\leq & 4 \maxz{v} \frac{\TVD(M)}{H\bigl( \Pr(v|M), F(v)\bigr)} 
\end{align}
which is equivalent to the upper bound presented in \autoref{eq:fixedMbound}.
The corresponding lower bound can be derived in a completely analogous fashion by starting off with $C_{a \to b}(F)$ instead of $C_{a \to b}(M)$. 
\end{proof}

This bound is not yet useful because the right hand side still depends on the unknown model. 
We seek an inequality that is independent of the model as long as the model has a fixed and sufficiently small value of $\TVD(M)$ with respect to the empirical frequencies. 
The bound in \autoref{sec:bounds} is a way to avoid this difficulty, and we have now assembled all prerequisites necessary to prove it. 
We restate the main theorem for completeness.

\addtocounter{theorem}{-1}
\begin{theorem}
Let $M$ denote any model and let $\TVD(M)$, $C_{a \to b} (M)$ and $C_{a \to b}(F)$ be as in \autoref{eq:TVD}, \autoref{eq:causal-influence-defn} and \autoref{eq:data_causal_effect}, respectively. 
Denote by $\mathcal{M}_\tau$ the set of models having $\TVD(M) \le \tau$ with respect to the empirical frequencies $F$, and let $f^\star = \min_{a,y} F(a,y)$, where $y$ denotes all parents  of the variable $b$ that are not grandparents of $b$. 
Then for all $M \in \mathcal{M}_\tau$ and $\tau < 2 f^\star$ we have
\begin{equation}
|C_{a \to b} (F) - C_{a \to b} (M)| \le \frac{2 \tau (4 f^\star-\tau )}{f^\star (2 f^\star-\tau )}\,.
\end{equation}
\end{theorem}

\begin{proof}[Proof of \autoref{thm:bounds}]
Again for the sake of bookkeeping we introduce a variable $v = (a,y)$. We begin with the inequality from \autoref{lem:fixedM} and note that we can simply maximize the righthand side over all $M \in \mathcal{M}_\tau$ to get a universal bound. 
We have 
\begin{equation}
	\max_{M\in \mathcal{M}_\tau}\frac{4 \; \TVD(M)}{\minz{v} H\bigl( \Pr(v|M),F(v)\bigr)} \le \frac{4 \; \tau}{\min_{M\in\mathcal{M}_\tau}\minz{v} H\bigl( \Pr(v|M),F(v)\bigr)} \,.
\end{equation}
Therefore we must establish a lower bound on the denominator. 
Plugging in the definition of the harmonic mean, a simple calculation confirms that $\partial_x H(x,y) = \frac{2 y^2}{(x+y)^2} \ge 0$, so the denominator is bounded from below as
\begin{align}\label{eq:denom1}
	\min_{M\in\mathcal{M}_\tau}\minz{v} \frac{2 \Pr(v|M) F(v)}{ \Pr(v|M) + F(v)}  \ge \min_{M\in\mathcal{M}_\tau}\minz{v} \frac{2 \Pr(v|M) f^\star}{ \Pr(v|M) + f^\star}\,.
\end{align}
Now we relax slightly to allow \emph{all} possible probability distributions (not necessarily ones coming from a causal model $M$, and denote $\mathcal{P}_\tau$ to be the set of all probability distributions $p$ with $\|p - F\|_1\le\tau$. Minimizing over a potentially larger set $\mathcal{P}_\tau$ may only decrease the function (or keep its minimum unchanged). We find the denominator is now bounded by
\begin{align}
	 \min_{M\in\mathcal{M}_\tau}\minz{v} \frac{2 \Pr(v|M) f^\star}{ \Pr(v|M) + f^\star} & \ge \min_{p\in\mathcal{P}_\tau}\minz{v} \frac{2 p(v) f^\star}{ p(v) + f^\star} \ge \min_{p\in\mathcal{P}_\tau} \frac{2 p^\star f^\star}{ p^\star + f^\star} \,.
\end{align}
Here in the second inequality we have used the same monotonicity argument for the harmonic mean above (since it is a symmetric function) and replaced the minimum over $v$ with $p^\star = \minz{v} p(v)$. 

Now we appeal to the monotonicity result of \autoref{lem:contraction}, so that $p\in\mathcal{P}_\tau$ implies that $\|p(v)-F(v)\|_1 \le \tau$. 
The claim then follows if we can establish the following result,
\begin{align}\label{eq:betterpstar}
	\min_{p\in\mathcal{P}_\tau} p^\star \ge f^\star -\frac{\tau}{2} \,.
\end{align}
A weaker result, that $\min_{p\in\mathcal{P}_\tau} p^\star \ge f^\star -\tau$, is easy to see if we relax the requirement that $p$ is normalized and add the more stringent requirement that $\tau < f^\star$. 
Begin with the choice $p(a) = F(a)$, and then subtract $\tau$ from the smallest component, keeping all other components fixed. 
This achieves the least value of this relaxed problem. 
This is a valid solution since the resulting vector is still nonnegative, owing to the constraint $\tau < f^\star$. 
The slightly tighter result follows from reasserting the constraint that the entries of $p$ must sum to 1, and allows us to weaken the constraint on $\tau$ to $\tau < 2 f^\star$. 
With the normalization condition in place, subtracting any deviation of size $\delta$ from a component of $p$ must be compensated by adding $\delta$ elsewhere in the vector, and this contributes a total of $2 \delta$ to the TVD between these differing vectors. 
The largest such a deviation can be is half of $\tau$, and to minimize our objective function we put this deviation on the smallest component. 
This component remains positive because of the condition $\tau < 2 f^\star$, so this remains a valid probability distribution.

Again by the monotonicity of the harmonic mean, this minimal value can be used to lower bound the denominator. The final inequality is obtained by plugging in the value of \autoref{eq:betterpstar} into the denominator expression and simplifying. 
\end{proof}

We remark that the maximum possible value for $f^\star$ in a Bell experiment where $a$ takes $d$ possible outcomes is $1/d$. 
Because this inequality is monotonically decreasing with $f^\star$, the bound becomes weaker as the number of outcomes increases, and the requirement that $\tau < 2 f^\star$ becomes more demanding.  

\end{document}